\documentclass[nofootinbib]{revtex4}
\usepackage{amsmath, amsfonts, graphicx, subfigure, amsthm}

\newtheorem{theorem}{Theorem}
\newtheorem{lemma}[theorem]{Lemma}

\newtheorem{corollary}[theorem]{Corollary}
\newcommand{\comment}[1]{}

\def\l{\left}
\def\r{\right}
\renewcommand{\>}{\rangle}
\newcommand{\<}{\langle}
\newcommand{\ket}[1]{|#1\rangle}
\newcommand{\bra}[1]{\langle #1|}
\newcommand{\braket}[2]{\langle #1|#2\rangle}
\newcommand{\proj}[1]{\left|#1\right\>\!\left\<#1\right|}
\newcommand{\ot}{\otimes}
\newcommand{\eps}{\epsilon}
\newcommand{\ra}{\rightarrow}
\def\benum{\begin{enumerate}}
\def\eenum{\end{enumerate}}
\def\bit{\begin{itemize}}
\def\eit{\end{itemize}}

\def\mn{\medskip\noindent}
\newcommand{\be}{\begin{equation}}
\newcommand{\ee}{\end{equation}}
\newcommand{\tr}{\mathop{\mathrm{tr}}\nolimits}
\def\ba#1\ea{\begin{align}#1\end{align}}
\def\ban#1\ean{\begin{align*}#1\end{align*}}
\newcommand{\nn}{\nonumber\\}
\newcommand{\eq}[1]{(\ref{eq:#1})}
\newcommand{\fig}[1]{Fig.~\ref{fig:#1}}
\newcommand{\secref}[1]{Section~\ref{sec:#1}}
\newcommand{\lemref}[1]{Lemma~\ref{lem:#1}}
\newcommand{\thm}[1]{Theorem~\ref{thm:#1}}
\newcommand{\cor}[1]{Corollary~\ref{cor:#1}}

\DeclareMathOperator{\cyc}{cyc}
\DeclareMathOperator{\Par}{Par}
\DeclareMathOperator{\poly}{poly}

\def\bbC{{\mathbb{C}}}
\def\bbE{{\mathbb{E}}}
\def\bbR{{\mathbb{R}}}
\def\cS{{\cal S}}
\def\vs{{\vec{s}}}
\def\vsi{{\vec{\sigma}}}
\def\cm{{c_2^{\max}}}

\begin{document}
\title{Random tensor theory: extending random matrix theory to
  mixtures of random product states}
\author{A.~Ambainis}
\affiliation{Faculty of Computing, University of Latvia, Riga, Latvia}
\author{A.~W.~Harrow}
\affiliation{Department of Mathematics, University of Bristol,
  Bristol, U.K.}
\author{M.~B.~Hastings}
\affiliation{Microsoft Research, Station Q, CNSI Building,
University of California, Santa Barbara, CA, 93106}
%\affiliation{Center for Nonlinear Studies and Theoretical Division,
%Los Alamos National Laboratory, Los Alamos, NM, 87545}

\begin{abstract}
  We consider a problem in random matrix theory that is inspired by
  quantum information theory: determining the largest eigenvalue of a
  sum of $p$ random product states in $(\bbC^d)^{\ot k}$, where $k$
  and $p/ d^k$ are fixed while $d\ra\infty$.  When $k=1$, the
  Mar\v{c}enko-Pastur law determines (up to small corrections) not
  only the largest eigenvalue ($(1+\sqrt{p/d^k})^2$) but the smallest
  eigenvalue $(\min(0,1-\sqrt{p/d^k})^2)$ and the spectral density in
  between.  We use the method of moments to show that for $k>1$ the
  largest eigenvalue is still approximately $(1+\sqrt{p/d^k})^2$ and
  the spectral density approaches that of the Mar\v{c}enko-Pastur law,
  generalizing the random matrix theory result to the random tensor
  case.  Our bound on the largest eigenvalue has implications both for
  sampling from a particular heavy-tailed distribution and for a
  recently proposed quantum data-hiding and correlation-locking scheme
  due to Leung and Winter.

  Since the matrices we consider have neither independent entries nor
  unitary invariance, we need to develop new techniques for their
  analysis.  The main contribution of this paper is to give three
  different methods for analyzing mixtures of random product states: a
  diagrammatic approach based on Gaussian integrals, a combinatorial
  method that looks at the cycle decompositions of permutations and a
  recursive method that uses a variant of the Schwinger-Dyson
  equations.
\end{abstract}
\maketitle

\def\D{d}

\section{Introduction and related work}
\subsection{Background}
A classical problem in probability is to throw $p$ balls into $\D$ bins
and to observe the maximum occupancy of any bin.  If we set the ratio
$x=p/\D$ to a constant and take $\D$ large, this maximum occupancy is
$O(\ln \D / \ln \ln \D)$ with high probability (in fact, this bound is
tight, but we will not discuss that here).  There are two natural ways
to prove this, which we call the large deviation method and the trace
method.  First, we describe the large deviation method.  If the
occupancies of the bins are $z_1,\ldots,z_\D$, then each $z_i$ is
distributed approximately according to a Poisson distribution with
parameter $x$; i.e.  $\Pr[z_i=z] \approx x^z/e^xz!$.  Choosing $z \gg
\ln \D / \ln \ln \D$ implies that $\Pr[z_i\geq z]\ll 1/\D$ for each $i$.
Thus, the union bound implies that with high probability all of the $z_i$ are
$\leq O(\ln d / \ln \ln d)$.  More generally, the large deviation
method proceeds by: (1) representing a bad event (here, maximum
occupancy being too large) as the union of many simpler bad events
(here, any one $z_i$ being too large), then (2) showing that each
individual bad event is very unlikely, and (3) using the union bound
to conclude that with high probability none of the bad events occur.
This method has been used with great success throughout classical and
quantum information theory~\cite{HLSW04,HLW06,Ledoux}. 

This paper will discuss a problem in quantum information theory where
the large deviation method fails.  We will show how instead a
technique called the trace method can be effectively used.  For the
problem of balls into bins, the trace method starts with the bound
$\max z_i^m \leq z_1^m + \ldots + z_d^m$, where $m$ is a large
positive integer.  Next, we take the expectation of both sides and use
convexity to show that
$$\bbE[\max z_i] \leq \l(\bbE[\max z_i^m]\r)^{\frac{1}{m}}
\leq \D^{\frac{1}{m}} \l(\bbE[z_1^m]\r)^{\frac{1}{m}}.$$ Choosing $m$
to minimize the right-hand side can be shown to yield the optimal $\ln
\D / \ln \ln \D + O(1)$ bound for the expected maximum occupancy.  In
general, this approach is tight up to the factor of $d^{1/m}$.

The quantum analogue of balls-into-bins problem is to choose $p$ random 
unit
vectors $\ket{\varphi_1},\ldots,\ket{\varphi_p}$ from $\bbC^\D$ and to
consider the spectrum of the matrix
\be M_{p,d} = \sum_{s=1}^p {\varphi_s}, \label{eq:M-def}\ee
where we use $\varphi$ to denote $\proj{\varphi}$.
Again we are interested in the regime where $x=p/\D$ is fixed and
$\D\ra\infty$.  We refer to this case as the ``normalized
ensemble.'' We also consider a slightly modified version of the
problem in which the states $\ket{\hat{\varphi}_s}$ are drawn from
a complex Gaussian distribution with unit variance, so that
the expectation of $\braket{\hat{\varphi}_s}{\hat{\varphi}_s}$ is equal to
one. Call the ensemble in the modified problem the ``Gaussian
ensemble'' and define $\hat{M}_{p,d}=  \sum_{s=1}^p \hat{\varphi}_s$.
Note that $\hat{M}_{p,d}=\hat\Phi^\dag \hat\Phi$, where $\hat \Phi = \sum_{s=1}^p \ket\varphi\bra s$ is a $p\times
\D$ matrix where each entry is an i.i.d. complex Gaussian variable
with variance $1/\D$.  That is, $\hat \Phi = \sum_{s=1}^p 
\sum_{j=1}^d (a_{s,j}
+ ib_{s,j})\ket{s}\bra{j}$, with $a_{s,j},b_{s,j}$ i.i.d. real
Gaussians each with mean zero and variance $1/2d$.

What we call the Gaussian ensemble is more conventionally
 known as the Wishart distribution, and
has been extensively studied.  Additionally, we will see in
\secref{gaussian} that the normalized ensemble is nearly the same as
the Gaussian ensemble for large $\D$.  In either version of the
quantum problem, the larger space from which we draw vectors means
fewer collisions than in the discrete classical case.  The nonzero
part of the spectrum of $M$ has been well studied\cite{Verbaar94a,Verbaar94b,FZ97}, and
it lies almost entirely between $(1\pm\sqrt{x})^2$ as $\D\ra\infty$. 
% (assuming $x\leq 1$).
This can be proven using a variety of techniques.  When $M$ is drawn
according to the Gaussian ensemble, its spectrum is described by
chiral random matrix theory\cite{Verbaar94a,Verbaar94b}.  This follows from the fact
that the spectrum of $M$ has the same distribution as the spectrum of
the square of the matrix \be
\label{eq:cm}
\begin{pmatrix}
0 & \hat\Phi \\
\hat\Phi^{\dagger} & 0
\end{pmatrix},
\ee
where $\hat\Phi$ is defined above.
 A variety of techniques have been used
to compute the spectrum\cite{FZ97,Anderson91a,Anderson91b,MP93}. %[and more?]  
The ability to use Dyson gas 
methods, or to perform exact integrals over the unitary group with a
Kazakov technique, has allowed even the detailed structure of the
eigenvalue spectrum near the edge to be worked out for this chiral
random matrix problem.   
%Finally, an elementary convexity argument can be used to show that
%eigenvalue bounds for the Gaussian ensemble imply bounds for the
%normalized ensemble, as we will see in \secref{gaussian}.

A large-deviation approach for the $x\ll 1$ case was given in
\cite[appendix B]{LW-locking}.  In order to bound the spectrum of
$M_{p,d}$, they instead studied the Gram matrix $M_{p,d}' := \Phi
\Phi^\dag$, which has the same spectrum as $M$.  Next they considered
$\bra{\phi}M_{p,d}'\ket{\phi}$ for a random choice of $\ket{\phi}$.
This quantity has expectation 1 and, by Levy's lemma, is within $\eps$
of its expectation with probability $\geq 1-\exp(O(\D\eps^2))$.  On
the other hand, $\ket{\phi}\in \bbC^p$, which can be covered by an
$\eps$-net of size $\exp(O(p\ln 1/\eps))$.  Thus the entire spectrum
of $M_{p,d}'$ (and equivalently $M_{p,d}$) will be contained in $1\pm
O(\eps)$ with high probability, where $\eps$ is a function of $x$ that
approaches 0 as $x\ra 0$.

In this paper, we consider a variant of the above quantum problem in
which none of the techniques described above is directly applicable.
We choose our states $\ket{\varphi_s}$ to be product states in
$(\bbC^\D)^{\ot k}$; i.e.
$$\ket{\varphi_s} = \ket{\varphi_s^1} \ot \ket{\varphi_s^2} \ot
 \cdots \ot \ket{\varphi_s^k},$$
for $\ket{\varphi_s^1},\ldots,\ket{\varphi_s^k}\in \bbC^\D$.  We choose
the individual states $\ket{\varphi_s^a}$ again either uniformly from
all unit vectors in $\bbC^\D$ (the normalized product ensemble) or as 
Gaussian-distributed  vectors with $\bbE
[\braket{\hat{\varphi}_s^a}{\hat{\varphi}_s^a}]=1$ (the Gaussian
product ensemble).  The corresponding matrices are $M_{p,d,k}$ and
$\hat{M}_{p,d,k}$ respectively.  Note 
that $k=1$ corresponds to the case considered above; i.e. 
$M_{p,d,1} =M_{p,d}$ and  $\hat{M}_{p,d,1} =\hat{M}_{p,d}$.  We are
interested in the case when $k>1$ is fixed.  As above, we also fix the
parameter 
$x=p/\D^k$, while we take $\D\ra\infty$.  
And as above, we would like to
show that the spectrum lies almost entirely within the region
$(1\pm\sqrt{x})^2$ with high probability.

However, the Dyson gas and Kazakov
techniques\cite{FZ97,Anderson91a,Anderson91b,MP93} that were used for
$k=1$ are not available for $k>1$, which may be considered a problem
of random tensor theory.  The difficulty is that we have a matrix with
non-i.i.d. entries and with unitary symmetry only within the $k$
subsystems.  Furthermore, large-deviation techniques are known to work
only in the $x \gg 1$ limits.  Here, Ref.~\cite{Rudelson} can prove
that $\|M - xI\| \leq O(\sqrt{xk \log(d)})$ with high probability,
which gives the right leading-order behavior only when $x\gg \sqrt{k
  \log d}$.  (The same bound is obtained with different techniques by
Ref.~\cite{AW02}.)  The case when $x\gg 1$ is handled by
Ref.~\cite{ALPT09a}, which can bound $\|M - xI \| \leq O(\sqrt{x})$
with high probability when $k\leq 2$.   (We will discuss this paper
further in \secref{convex}.)
%An unpublished conjecture of G.~Aubrun would extend the techniques of Ref.~\cite{ALPT09a} to prove this bound for general $k$ (although still with the restriction that $x\gg 1$.)

However, some new techniques will be needed to cover the case when $x\leq O(1)$.  Fortunately it turns out that the diagrammatic techniques for
$k=1$ can be modified to work for general $k$.  In \secref{matt}, we will use these
techniques to obtain an expansion in $1/\D$.  Second, the large
deviation approach of \cite{LW-locking} achieves a concentration bound of
$\exp(-O(\D\eps^2))$ which needs to overcome an $\eps$-net of size
$\exp(O(p\ln(1/\eps)))$.  This only functions when $p\ll \D$, but we
would like to take $p$ nearly as large as $\D^k$.  One approach when
$k=2$ is to use the fact that $\bra{\psi}M_{p,d,k}\ket{\psi}$ exhibits smaller
fluctuations when $\ket{\psi}$ is more entangled, and that most states
are highly entangled.  This technique was used in an unpublished
manuscript of Ambainis to prove that $\|M_{p,d,k}\|=O(1)$ with high probability when
$p=O(\D^2/\poly\ln(\D))$. However, the methods in this paper are
simpler, more general and achieve stronger bounds.

Our strategy to bound the typical value of the largest eigenvalue of
$M_{p,d,k}$ will be to use a trace method: we bound the expectation value 
of
the trace of a high power, denoted $m$, of $M_{p,d,k}$.  This yields an 
upper
bound on $\|M_{p,d,k}\|$ because of the following key inequality 
\be
\label{eq:tm}
\|M_{p,d,k}\|^m \leq 
{\rm tr}(M_{p,d,k}^m).
\ee
We then proceed to expand $\bbE[\tr M_{p,d,k}^m]$ (which we denote
$E_{p,d,k}^m$) as
\be
\label{eq:sum}
E_{p,d,k}^m = 
\bbE[\tr M_{p,d,k}^m] = 
\sum_{s_1=1}^p \sum_{s_2=1}^p ... \sum_{s_m=1}^p
E_d[s_1,s_2,...,s_m]^k
: = \sum_{\vs\in [p]^m} E_d[\vs]^k
\ee
where
\be
E_d[\vs]= \bbE\l[{\rm tr}(\varphi_{s_1} \varphi_{s_2}
\ldots \varphi_{s_m})\r],
\ee
and $[p]=\{1,\ldots,p\}$.
Similarly we define $\hat{E}_{p,d,k}^m = \bbE[\tr \hat{M}_{p,d,k}^M]$
and $\hat{E}_d[\vs]= \bbE\l[{\rm tr}(\hat\varphi_{s_1} \hat\varphi_{s_2}
\ldots \hat\varphi_{s_m})\r]$, and observe that they obey a relation
analogous to \eq{sum}.  We also define the normalized traces
$e_{p,d,k}^m = d^{-k}E_{p,d,k}^m$ and 
$\hat{e}_{p,d,k}^m = d^{-k}\hat{E}_{p,d,k}^m$, which will be useful
for understanding the eigenvalue density.

The rest of the paper presents three independent proofs that for
appropriate choices of $m$,  $E_{p,d,k}^m = (1 +
\sqrt{x})^{2m}\exp(\pm o(m))$.  This will  
imply that $\bbE[\|M_{p,d,k}\|] \leq (1 + \sqrt{x})^2 \pm o(1)$, which
we can combine with standard measure concentration results to give
tight bounds on the  probability that
$\|M_{p,d,k}\|$ is far from $(1+\sqrt{x})^2$.  We will also derive
nearly matching lower bounds on $E_{p,d,k}^m$ which show us that the
limiting spectral density of $M_{p,d,k}$ matches that of the Wishart
distribution (a.k.a. the $k=1$ case).
  The reason for the multiple proofs is to introduce new
techniques to problems in quantum information that are out of reach of
the previously used tools.  The large-deviation techniques used for
the $k=1$ case have had widely successful applicability to quantum
information and we hope that the methods introduced in this paper will
be useful in the further exploration of random quantum states and
processes.  Such random states, unitaries, and measurements play an important role in many
area of quantum information such as encoding quantum\cite{ADHW06},
classical\cite{HW-additivity,Hastings-additivity}, and
private\cite{private-super} 
information over quantum channels,
in other data-hiding schemes\cite{HLSW04}, in quantum
expanders\cite{Hastings-expander1,BST08-expander}, and 
in general coding protocols\cite{YD07}, among other applications.

The first proof, in \secref{matt}, first uses the expectation over the
Gaussian ensemble to upper-bound the expectation over the normalized
ensemble. Next, it uses Wick's theorem to give a diagrammatic method
for calculating the expectations.  A particular class of diagrams,
called rainbow diagrams, are seen to give the leading order terms.
Their contributions to the expectation can be calculated exactly,
while for $m\ll \D^{1/2k}$, the terms from non-rainbow diagrams are
shown to be negligible.  In fact, if we define the 
generating function
\be \hat{G}(x,y) = \sum_{m\geq 0} y^m \hat{e}_{p,d,k}^m,
\label{eq:mom-gen-fun}\ee
then the methods of \secref{matt} can be used to calculate
\eq{mom-gen-fun} up to $1/\D$ corrections.  Taking the analytic
continuation of $G(x,y)$ gives an estimate of the eigenvalue
density across the entire spectrum of $M_{p,d,k}$.  More precisely,
since we can only calculate the generating function up to $1/d$ corrections,
we can use convergence in moments to show that the distribution of eigenvalues
weakly converges almost surely (\cor{measureconverge} below)
to a limiting distribution.  For this limiting distribution,
for $x<1$, the eigenvalue density of $M_{p,d,k}$ vanishes for
eigenvalues less than
$(1-\sqrt{x})^2$.  However, this calculation, in contrast to the
calculation of the largest eigenvalue, only tells us that the fraction
of eigenvalues outside $(1\pm\sqrt{x})^2$ approaches zero with high
probability, and cannot rule out the existence of a small number of
low eigenvalues.

The second proof, in \secref{aram}, is based on representation theory
and combinatorics.  It first repeatedly applies two simplification
rules to $E_d[\vs]$: replacing occurrences of
$\varphi_s^2$ with $\varphi_s$ and replacing $\bbE[\varphi_s]$ with $I/\D$
whenever $\varphi_s$ appears only a single time in a string.  Thus
$\vs$ is replaced by a (possibly empty) string $\vs'$ with no repeated
characters and with no characters occurring only a single time.  To
analyze $E_d[\vs']$, we express $\bbE[\varphi^{\ot n}]$ as a sum over
permutations and use elementary arguments to enumerate permutations
with a given number of cycles.  We find that the dominant contribution
(corresponding to rainbow diagrams from \secref{matt}) comes from the
case when $\vs'=\emptyset$, and also analyze the next leading-order
contribution, corresponding to $\vs'$ of the form 1212, 123213, 12343214, 1234543215, etc.
Thus we obtain an estimate for $E_{p,d,k}^m$ that is correct up to an
$o(1)$ additive approximation.

The third proof, in \secref{andris}, uses the Schwinger-Dyson
equations to remove one letter at a time from the string $\vs$.  This
leads to a simple recursive formula for $e_{p,d,k}^m$ that gives precise
estimates.  

All three proof techniques can be used to produce explicit
calculations of $E_{p,d,k}^m$.  Applying them for the first few values of $m$ yields
\begin{eqnarray*} E_{p,d,k}^1 &= &p \\
E_{p,d,k}^2 &= & p + \frac{(p)_2}{d^k} \\
E_{p,d,k}^3 &= & p + 3\frac{(p)_2}{d^k} +\frac{(p)_3}{d^{2k}} \\
E_{p,d,k}^4 &= & p + 6\frac{(p)_2}{d^k} + 6\frac{(p)_3}{d^{2k}}
+\frac{(p)_4}{d^{3k}} + 2^k\frac{(p)_2}{d^{k}(d+1)^k}\\
E_{p,d,k}^5 &= &  p + 10\frac{(p)_2}{d^k} + 20\frac{(p)_3}{d^{2k}}
+10\frac{(p)_4}{d^{3k}}+\frac{(p)_5}{d^{4k}} + 5\cdot 2^k\frac{(p)_2}{d^{k}(d+1)^k}\\
E_{p,d,k}^6 &= &  p + 15\frac{(p)_2}{d^k} + 50\frac{(p)_3}{d^{2k}}
+50\frac{(p)_4}{d^{3k}}+15\frac{(p)_5}{d^{4k}}+\frac{(p)_6}{d^{5k}} 
\\&& + 15\cdot 2^k\frac{(p)_2}{d^{k}(d+1)^k} 
+\frac{(p)_2(d+3)^k}{d^{2k}(d+1)^{2k}}
+ 6^k \frac{(p)_3}{d^k(d+1)^k(d+2)^k}
,\end{eqnarray*}
where $(p)_t = p!/(p-t)! = p(p-1)\cdots (p-t+1)$.
We see that $O(1)$ (instead of $O(d^k)$) terms start to appear when
$m\geq 4$.  The combinatorial significance of these will be discussed 
in \secref{irred-strings}.

\subsection{Statement of results}
Our main result is the following theorem.
\begin{theorem}\label{thm:trace}
Let $\beta_m(x) = \sum_{\ell=1}^m N(m,\ell)x^\ell,$
where $N(m,\ell)=\frac{1}{m}\binom{m}{\ell-1}\binom{m}{\ell}$ are known
as the Narayana numbers. Then,
\be \l(1-\frac{m^2}{p}\r) \beta_m\!\l(\frac{p}{d^k}\r)\leq
 \frac{1}{d^k}\bbE[\tr(M_{p,d,k}^m)] \leq
 \exp\l(\frac{3m^{k+4}}{xd^{1/k}}\r)\beta_m\!\l(\frac{p}{d^k}\r),
 \label{eq:main-e-bound}\ee
where $\exp(A):=e^A$ and the lower bound holds only when $m<\sqrt{p}$.

Thus, for all $m\geq 1$, $k\geq 1$, $x>0$ and $p=x d^k$,
$$\lim_{d\ra \infty} e_{p,d,k}^m =\beta_m(p/d^k),$$
where we have used the notation
$e_{p,d,k}^m  = \frac{1}{d^k}\bbE[\tr(M_{p,d,k}^m)]$.
\end{theorem}

Variants of the upper bound are proven separately in each of the next three
sections, but the formulation used in the Theorem is proven in \secref{andris}.
Since the lower bound is simpler to establish, we prove it only in
\secref{aram}, although the techniques of 
Sections~\ref{sec:matt} and \ref{sec:andris} would also give nearly the same bound.

For the data-hiding and correlation-locking scheme proposed in \cite{LW-locking}, it is important that $\|M\|=1 + o(1)$ whenever $x=o(1)$.   In fact, we will show that $\|M\|$ is very likely to be close to $(1+\sqrt{x})^2$, just as was previously known for Wishart matrices.  First we observe that for
 large $m$, $\beta_m(x)$ is roughly $(1+\sqrt{x})^{2m}$. % = \lambda_+^m$.

\begin{lemma}\label{lem:beta-comp}
\be \frac{x}{2m^2(1+\sqrt{x})^3} (1+\sqrt{x})^{2m}
 \leq \beta_m(x) \leq (1+\sqrt{x})^{2m}\ee
\end{lemma}

The proof is deferred to \secref{large-dev}.

Taking $m$ as large as possible in \thm{trace} gives us tight
bounds on the typical 
behavior of $\|M_{p,d,k}\|$.
\begin{corollary}\label{cor:eig-ub}
With $M_{p,d,k}$ and $x$ defined as above,
$$(1+\sqrt{x})^2 - O\l(\frac{\ln(d)}{\sqrt{d}}\r) \leq\bbE[\|M_{p,d,k}\|] 
\leq (1+\sqrt{x})^2 + O\l(\frac{\ln(\D)}{\D^{\frac{1}{2k}}}\r)$$
and the same bounds hold with $\ M_{p,d,k}$ replaced by $\hat M_{p,d,k}$.
\end{corollary}
\begin{proof}
A weaker version of the upper bound can be established by setting $m\sim d^{1/k(k+4)}$ in
\be \bbE[\|M_{p,d,k}\|] \leq 
\bbE[\|M_{p,d,k}\|^m]^{1/m} \leq
d^{\frac{k}{m}} (e_{p,d,k}^m)^{\frac{1}{m}},
\label{eq:eig-from-trace}\ee where the first
inequality is from the convexity of $x\mapsto x^m$.  In fact, the version stated here is proven in \eq{diag-eig-bound} at the end of \secref{matt}.

The lower bound will be proven in \secref{large-dev}.
\end{proof}

Next, the reason we can focus our analysis on the expected value of $\|M_{p,d,k}\|$ is because $\|M_{p,d,k}\|$ is extremely unlikely to be far from its mean.  Using standard measure-concentration arguments (detailed in \secref{large-dev}), we can prove:
\begin{lemma}\label{lem:large-dev}
For any $\eps>0$,
\be \Pr\l(\l|\|M_{p,d,k}\| - \bbE[\|M_{p,d,k}\|]\r| \geq \eps\r) \leq
2 \exp(-(d-1)\eps^2/k).
\label{eq:norm-LD}\ee
For any $0<\eps\leq 1$,
\be \Pr\l(\l|\|\hat M_{p,d,k}\| - \bbE[\|M_{p,d,k}\|]\r| \geq 
\eps\bbE[\|M_{p,d,k}\|] + \delta \r) \leq
2pke^{-\frac{d\eps^2}{4k^2}} + 2e^{-\frac{(d-1)\delta^2}{4k}}
\label{eq:gaussian-LD}\ee
\end{lemma}
Combined with \cor{eig-ub} we obtain:
\begin{corollary}\label{cor:eig-LD}
$$\Pr\l(\l|\|M_{p,d,k}\| - \lambda_+]\r| \geq O\l(\frac{\ln
  \D}{\D^{1/2k}}\r) + \eps\r) \leq
2 \exp(-d\eps^2/2).$$
A similar, but more cumbersome, bound also exists for $\|\hat
M_{p,d,k}\|$.
\end{corollary}

Note that for the $k=1$ case, the exponent can be replaced by
$O(-d\eps^{3/2})$, corresponding to typical fluctuations on the order
of $O(d^{-2/3})$ \cite{Johnstone-PCA}. It is plausible that fluctuations of
this size would also hold in the $k>1$ case as well, but we do not
attempt to prove that in this paper.

Our asymptotic estimates for $e_{p,d,k}^m$ also imply that the
limiting spectral density of $M_{p,d,k}$ is given by the
Mar\v{c}enko-Pastur law, just as was previously known for the $k=1$
case.  Specifically, let $\lambda_1,\ldots,\lambda_{R}$ be the non-zero eigenvalues of $M_{p,d,k}$, with $R = \text{rank} M_{p,d,k}$.  Generically $R = \min(p, d^k)$ and the eigenvalues are all distinct.  Define the eigenvalue density to be 
$$\rho(\lambda)= \frac{1}{R} \sum_{i=1}^{R} \delta(\lambda_i-
\lambda),$$
 then
\begin{corollary}
\label{cor:measureconverge}
In the limit of large $d$ at fixed $x$, $\rho(\lambda)$ weakly
converges almost surely to 
$$
\frac{\sqrt{(\lambda_+-\lambda)(\lambda-\lambda_-)}}{2\pi x
    \lambda }I(\lambda_- \leq \lambda \leq \lambda_+)$$
    for any fixed $k$ and for both the normalized and Gaussian ensembles.
\end{corollary}
Here $\lambda_{\pm}= (1\pm\sqrt{x})^2$ and $I(\lambda_- \leq \lambda
\leq \lambda_+)=1$ if $\lambda_- \leq \lambda
\leq \lambda_+$ and 0 otherwise.

This corollary follows from \thm{trace} using standard
arguments\cite{another-moment}.  We believe, but are unable to prove,
that in the $x\leq 1$ case, 
the probability of
any non-zero eigenvalues existing below $\lambda_- - \eps$ vanishes for any
$\eps>0$ 
in the limit of large $d$ at fixed $x$, just as is known when $k=1$.

\subsection{Applications}
\subsubsection{Data hiding}

One of the main motivations for this paper was to analyze the proposed data-hiding and correlation-locking scheme of \cite{LW-locking}.  In this section, we will briefly review their scheme and 
explain the applicability of our results.  

Suppose that $p = d
\log^c(d)$ for some constant $c>0$, and we consider the $k$-party
state $\rho = \frac{1}{p}\sum_{s=1}^p \varphi_s$.  We can think of $s$
as a message of $(1+o(1))\log d$ bits that is ``locked'' in the shared
state.  In \cite{LW-locking} it was proved that any LOCC (local operations and classical communication) protocol that uses a
constant number of rounds cannot produce an output with a
non-negligible amount of mutual information with $s$, and \cite{LW-locking} also proved that
the parties cannot recover a non-negligible amount of mutual information with each other that would not be
revealed to an eavesdropper on their classical communication so that the state cannot be used to
produce a secret key.  (They also
conjecture that the same bounds hold for an unlimited number of
rounds.)   However, if $c\log\log(d) + \log(1/\eps)$ bits of $s$ are
revealed then each party is left with an unknown state from a set of
$\eps d$ states in $d$ dimensions. Since these states are randomly
chosen, it is possible for each party to correctly identify the
remaining bits of $s$ with probability $1-O(\eps)$~\cite{Ashley-distinguish}.

On the other hand, the bounds on the eigenvalues of $\rho$ established by our
\cor{eig-LD} imply that the scheme of Ref.~\cite{LW-locking} can
be broken by a separable-once-removed  quantum measurement\footnote{This refers to a POVM (positive operator valued measure) in which all but one of
the measurement operators are product operators.}: specifically the
measurement given by completing $\{\frac{p}{\|\rho\|}\varphi_s\}_s$
into a valid POVM.  We hope that our bounds will also be of
use in proving their conjecture about LOCC distinguishability with an
unbounded number of rounds.  If this conjecture is established then it
will imply a dramatic separation between the strengths of LOCC and
separable-once-removed quantum operations, and perhaps could be strengthened
to separate the strengths of LOCC and separable operations.

\subsubsection{Sampling from heavy-tailed distributions}
\label{sec:convex}
A second application of our result is to convex geometry.
The matrix $M$ can be thought of as the empirical
covariance matrix of a collection of random product vectors.
These random product vectors have unit norm, and the distribution has $\psi_r$ norm on the order of $1/\sqrt{d^k}$ iff $r$ satisfies $r\leq 2/k$. 
 Here the $\psi_r$ norm is defined  (following \cite{ALPT09b}) for $r>0$ and for a scalar random variable $X$ as
\be  \| X \|_{\psi_r} = \inf \{C > 0 : \bbE[ \exp (|X| / C)^r] \leq
2\}
\label{eq:psir-def}\ee
and for a random vector $\ket\varphi$ is defined in terms of its linear forms:
$$ \| \varphi \|_{\psi_r} = \sup_{\ket\alpha} \| \braket{\alpha}{\varphi} \|_{\psi_r},$$
where the $\sup$ is taken over all unit vectors $\alpha$.
Technically, the $\psi_r$ norm is not a norm for $r<1$, as it does
not satisfy the triangle inequality.  To work with an actual norm, we
could replace \eq{psir-def} with $\sup_{t \geq 1} \bbE[|X|^t]^{1/t} /
t^{1/r}$, which similarly captures the tail dependence.   We mention
also that the $\psi_2$ norm has been called the subgaussian moment and
the $\psi_1$ norm the subexponential moment.

Thm.~3.6 of Ref.~\cite{ALPT09a} proved that when $M$ is a sum of vectors from a distribution with bounded $\psi_1$ norm and $x\gg 1$ then $M$ is within $O(\sqrt{x}\log x)$ of  $xI$ with high probability.
And as we have stated, Refs.~\cite{Rudelson,AW02} can prove
that $\|M - xI\| \leq O(\sqrt{xk \log(d)})$ with high probability, even without assumptions on $r$, although the bound is only meaningful when $x\gg \sqrt{k
  \log d}$.
In the case when $x\ll 1$ and $1\leq r\leq 2$ (i.e. $k\leq 2$), Thm 3.3 of Ref.~\cite{ALPT09b} proved that $M$ is within $O(\sqrt{x}\log^{1/r}(1/x))$ of a rank-$p$ projector.   Aubrun has conjectured that their results should hold for $r>0$ and any distribution on $D$-dimensional unit vectors with  $\psi_r$ norm $\leq O(1/\sqrt{D})$. If true, this would cover the ensembles that we consider.

Thus, our main result bounds the spectrum of $M$ in a
setting that is both more general than that of \cite{ALPT09a,ALPT09b} (since we allow general $x>0$ and $k\geq 1$, implying that $r=2/k$ can be arbitrarily close to 0) and more specific (since we do not consider only products of uniform random vectors, and not general ensembles with bounded $\psi_r$ norm).  
Our results can be viewed as evidence in support of Aubrun's conjecture.

\subsection{Notation}
For the reader's convenience, we collect here the notation used
throughout the paper.   This section omits variables that are used
only in the section where they are defined.

\mn
\begin{tabular}{|l|l|}
\hline
Variable & Definition \\ \hline
$d$ &  local dimension of each subsystem.\\
$k$ &  number of subsystems.  \\
$p$ &  number of random product states chosen\\
$x$ &  $p/d^k$. \\
$\ket{\varphi_s^i}$ & unit vector chosen at random from $\bbC^d$
for $s=1,\ldots,p$ and $i=1,\ldots,k$. \\
$\ket{\hat{\varphi_s^i}}$ & Gaussian vector from $\bbC^d$ with
$\bbE[\braket{\varphi_s^i}{\varphi_s^i}]=1$. \\
$\varphi$ & $\proj{\varphi}$ (for any state $\ket{\varphi}$)\\
$\ket{\varphi_s}$ & $\ket{\varphi_s^1} \ot \cdots \ot
\ket{\varphi_s^k}$\\
$\ket{\hat\varphi_s}$ & $\ket{\hat\varphi_s^1} \ot \cdots \ot
\ket{\hat\varphi_s^k}$\\
$M_{p,d,k}$ & $\sum_{s=1}^p \varphi_s$ \\
$\lambda_{\pm}$ & $ (1\pm\sqrt{x})^2$ \\
%$\hat M_{p,d,k}$ & $\sum_{s=1}^p \hat\varphi_s$ \\
$E_{p,d,k}^m$ & $\bbE[\tr M_{p,d,k}^m]$ \\
%$\hat E_{p,d,k}^m$ & $\bbE[\tr \hat M_{p,d,k}^m]$ \\
$e_{p,d,k}^m$ & $\frac{1}{d^k}\bbE[\tr M_{p,d,k}^m]$ \\
%$\hat e_{p,d,k}^m$ & $\frac{1}{d^k}\bbE[\tr \hat M_{p,d,k}^m]$ \\
$E_d[\vs]$ & $\bbE[\tr(\varphi_{s_1}\cdots \varphi_{s_m})]$, where 
$\vs = (s_1,\ldots,s_m)$ \\
$G(x,y)$ & $\sum_{m\geq 0} y^m e_{p,d,k}^m$ \\
$\beta(x)$ & $\sum_{\ell=1}^m N(m,\ell) x^\ell$ \\
$N(m,\ell)$ & Narayana number: $\frac{1}{m}\binom{m}{\ell-1}\binom{m}{\ell} = 
\frac{1}{\ell}
\binom{m}{\ell-1}\binom{m-1}{\ell-1}=\frac{m!m-1!}{\ell!\ell-1!m-\ell!m-\ell+1!}$
(and $N(0,0)=1$)\\
$F(x,y)$ & $\sum_{0\leq\ell\leq m<\infty} N(m,\ell) x^\ell y^m$\\
\hline
\end{tabular}

\mn
We also define $\ket{\hat\varphi_s}$, $\hat{M}_{p,d,k}$, $\hat
E_{p,d,k}$, $\hat{G}(x,y)$ and so on by replacing $\ket{\varphi_s^i}$
with $\ket{\hat\varphi_s^i}$.

\subsection{Proof of large deviation bounds}
\label{sec:large-dev}

In this section we prove \lemref{beta-comp}, \lemref{large-dev} and the lower bound of
\cor{eig-ub}.  First we review some terminology and basic results from
large deviation theory, following Ref.~\cite{Ledoux}.  Consider a set $X$
with an associated measure $\mu$ and distance metric $D$.  If
$Y\subseteq X$ and $x\in X$ then define $D(x,Y) := \inf_{y\in Y}
D(x,y)$.  For any $\eps\geq 0$ define $Y_\eps := \{x\in X: D(x,Y)\leq
\eps\}$.  Now define the concentration function $\alpha_X(\eps)$ for
$\eps\geq 0$ to be
$$\alpha_X(\eps) := \max \{ 1 - \mu(Y_\eps) : \mu(Y) \geq 1/2\}.$$
Say that $f:X\ra \bbR$ is $\eta$-Lipschitz if $|f(x)-f(y)| \leq \eta
D(x,y)$ for any $x,y\in X$.  If $m$ is a median value of $f$
(i.e. $\mu(\{x : f(x)\leq m\})=1/2$) then we can combine these
definitions to obtain the concentration result \be \mu(\{x: f(x) \geq
m + \eta \eps\}) \leq \alpha_X(\eps)
\label{eq:gen-concentration}.\ee
Proposition 1.7 of Ref.~\cite{Ledoux} proves that \eq{gen-concentration}
also holds when we take $m=\bbE_{\mu}[f]$.

Typically we should think of $\alpha_X(\eps)$ as decreasingly
exponentially with $\eps$.  For example, Thm 2.3 of \cite{Ledoux}
proves that $\alpha_{S^{2d-1}}(\eps) \leq e^{-(d-1)\eps^2}$, where
$S^{2d-1}$ denotes the unit sphere in $\bbR^{2d}$, $\mu$ is the
uniform measure and we are using the Euclidean distance.

To analyze independent random choices, we define the $\ell_1$ direct
product $X^n_{\ell^1}$ to be the set of $n$-tuples $(x_1,\ldots,x_n)$
with distance measure $D_{\ell^1}((x_1,\ldots,x_n),(y_1,\ldots,y_n))
:= D(x_1,y_1) + \ldots + D(x_n,y_n)$.  Similarly define $X^n_{\ell^2}$
to have distance measure
$D_{\ell^2}((x_1,\ldots,x_n),(y_1,\ldots,y_n)) := \sqrt{D(x_1,y_1)^2 +
  \ldots + D(x_n,y_n)^2}$.

Now we consider the normalized ensemble.  Our random matrices are
generated by taking $pk$ independent draws from $S^{2d-1}$,
interpreting them as elements of $\bbC^{d}$ and then constructing
$M_{p,d,k}$ from them.  We will model this as the space
$((S^{2d-1})^p_{\ell^2})^k_{\ell^1}$.  First, observe that Thm 2.4 of
\cite{Ledoux} establishes that
 $$\alpha_{(S^{2d-1})^p_{\ell^2}}(\eps) \leq e^{-(d-1)\eps^2} .$$
Next, Propositions 1.14 and 1.15 of \cite{Ledoux} imply that 
$$\alpha_{((S^{2d-1})^p_{\ell^2})^k_{\ell^1}}(\eps) \leq e^{-(d-1)\eps^2/k} .$$
Now we consider the map $f:((S^{2d-1})^p_{\ell^2})^k_{\ell^1} \ra
\bbR$ that is defined by $f(\{ \ket{\varphi_s^i}\}_{s=1,\ldots,p \atop
  i=1,\ldots,k}) = \| M_{p,d,k}\|$, with $M_{p,d,k}$ defined as usual
as $M_{p,d,k} = \sum_{s=1}^p \varphi_s^1 \ot \cdots \ot \varphi_s^k$.
To analyze the Lipschitz constant of $f$, note that the function $M\ra
\|M\|$ is 1-Lipschitz if we use the $\ell_2$ norm for matrices
(i.e. $D(A,B) = \sqrt{\tr(A-B)^\dag(A-B)}$)~\cite{HJ}.  Next, we can
use the triangle inequality to show that the defining map from
$((S^{2d-1})^p_{\ell^2})^k_{\ell^1}$ to $M_{p,d,k}$ is also
1-Lipschitz.  Thus, $f$ is 1-Lipschitz.  Putting this together we
obtain the proof of \eq{norm-LD}.

Next, consider the Gaussian ensemble.  any Gaussian vector
$\ket{\hat\varphi_s^i}$ can be expressed as $\ket{\hat\varphi_s^i} =
\sqrt{r_{s,i}}\ket{\varphi_s^i}$, where $\ket{\varphi_s^i}$ is a
random unit vector in $\bbC^d$ and $r_{s,i}$ is distributed according to
$\chi_{2d}^2/2d$. Here $\chi_{2d}^2$ denotes the chi-squared
distribution with $2d$ degrees of freedom; i.e. the sum of the squares
of $2d$ independent Gaussians each with unit variance.

The normalization factors are extremely likely to be close to 1.
First, for any $t<d$ one can compute
$$\bbE[e^{tr_{s,i}}] = (1-t/d)^{-d}.$$
Combining this with Markov's inequality implies that 
$\Pr[r_{s,i}\geq 1+\eps] = \Pr[e^{tr_{s,i}}\geq e^{t(1+\eps)}]  \leq 
(1-t/d)^{-d}e^{-t(1+\eps)}$ for any $t>0$.  We will set $t=d\eps/(1+\eps)$ and
then find that
\be \Pr[r_{s,i}\geq 1+\eps] \leq e^{-d(\eps-\ln(1+\eps))}
\leq e^{-\frac{d\eps^2}{4}},
\label{eq:csi-ub}\ee
where the second inequality holds when $\eps \leq 1$.
Similarly we can take $t=-d\eps/(1-\eps)$ 
to show that
$$\Pr[r_{s,i}\leq 1-\eps] \leq e^{d(\eps + \ln(1-\eps))}
\leq e^{-\frac{d\eps^2}{2}}.$$
Now we use the union bound to argue that with high probability {
  none} of the $r_{s,i}$ are far from 1.  In particular the
probability that {\em any} $r_{s,i}$ differs from 1 by more than $\eps/k$ is
$\leq 2pke^{-d\eps^2/4 k^2}$.

In the case that all the $r_{s,i}$ are close to 1, we can then obtain the
operator inequalities
\be (1-\eps)M_{p,d,k}\leq \hat{M}_{p,d,k} \leq (1+2\eps)M_{p,d,k}.
\label{eq:M-op-ineq}\ee
(For the upper bound we use $(1+\eps/k)^k \leq e^\eps \leq 1+2\eps$
for $\eps\leq 1$.)
This establishes that $\|\hat{M}_{p,d,k}\|$ is concentrated around the
expectation of  
$\bbE[\|{M}_{p,d,k}\|]$, as claimed in \eq{gaussian-LD}.

One application of these large deviation bounds is to prove the lower
bound in \cor{eig-ub}, namely that $(1+\sqrt{x})^2 - O\l(\frac{\ln
  \D}{\D^{1/2k}}\r) \leq\bbE[\|M_{p,d,k}\|]$.  First observe that
\thm{trace} and \lemref{beta-comp} imply that
$$\frac{\l(1-\frac{m^2}{p}\r)x}{2m^2(1+\sqrt{x})^3} \lambda_+^m \leq
e_{p,d,k}^m.$$ 
On the other hand, $d^{-k}\tr M \leq \|M\|$ and so $e_{p,d,k}^m \leq
\bbE[\|M_{p,d,k}\|^m]$.   Define $\mu := \bbE[\|M_{p,d,k}\|]$.  Then
\ban e_{p,d,k}^m & \leq \bbE[\|M_{p,d,k}\|^m] \\
& = \int_0^\infty d\lambda \Pr[\|M_{p,d,k}\| \geq \lambda] m
\lambda^{m-1}
& \text{using integration by parts} \\
& \leq \mu^m + m\int_0^\infty d\eps (\mu + \eps)^{m-1} \Pr[\|M_{p,d,k}\|
\geq \mu + \eps] \\
& \leq \mu^m + m\int_0^\infty d\eps (\mu + \eps)^{m-1} 
e^{-\frac{(d-1)\eps^2}{k}}
& \text{from \eq{norm-LD}} 
\\ & \leq \mu^m\l(1 + m %\frac{m}{\mu}
\int_0^\infty d\eps 
\exp\l((m-1)\eps%\frac{\eps}{\mu}
-\frac{(d-1)\eps^2}{k}\r)\r)
& \text{using $1+\eps/\mu \leq 1+\eps \leq e^\eps$}
\\ & \leq \mu^m\l(1 + m %\frac{m}{\mu}
\int_{-\infty}^\infty d\eps 
\exp\l(-\frac{d-1}{k}\l(\eps - \frac{k(m-1)}{2(d-1)}\r)^2
 + \frac{k^2(m-1)^2}{4(d-1)}\r)\r)
& \text{completing the square} 
\\ & \leq \mu^m\l(1 + m%\frac{m}{\mu}
\sqrt{\frac{2\pi k}{d-1}}
\exp\l(\frac{k^2(m-1)^2}{4(d-1)}\r)\r)
& \text{performing the Gaussian integral}
\ean
Combining these bounds on $e_{p,d,k}^m$ and taking the $m^{\text{th}}$
root we find that
$$\mu \geq \lambda_+ \l(\frac{\l(1-\frac{m^2}{p}\r)x}
{2m^2(1+\sqrt{x})^3\l(1  + m\sqrt{\frac{2\pi k}{d-1}}
\exp\l(\frac{k^2(m-1)^2}{4(d-1)}\r)\r)}\r)^{\frac{1}{m}}$$
Assuming that $m^2 \leq p/2$ and $m^2k^2 \leq d$, we find that 
$\mu \geq \lambda_+\l(1 - O\l(\frac{\ln(m)}{m}\r)\r)
= \lambda_+\l(1 - O\l(\frac{\ln(d)}{\sqrt{d}}\r)\r)$, which yields
the lower bound on $\bbE[\|M\|]$ stated in \cor{eig-ub}.  We omit the similar, but more tedious, arguments that can be used to lower-bound $\bbE[\|\hat M\|]$.

We conclude the section with the proof of \lemref{beta-comp}.
\begin{proof}
For the upper bound, note that $N(m,\ell)\leq \binom{m}{\ell}^2 \leq
\binom{2m}{2\ell}$ and so 
$$\sum_{\ell=1}^m N(m,\ell) x^\ell \leq
\sum_{\ell'=2}^{2m} \binom{2m}{\ell'} \sqrt{x}^{\ell'} 
 = (1 + \sqrt{x})^{2m}.$$

For the lower bound, first observe that
$$\frac{\binom{2m}{2\ell}}{\binom{m}{\ell}^2} = \frac{2m(2m-1)\cdots
  (2m-2\ell+1)}{m\cdot m\cdot (m-1)\cdot(m-1)\cdots(m-\ell+1)
\cdot(m-\ell+1)} \cdot \frac{\ell!^2}{(2\ell)!} \leq
2^{2\ell}\cdot \frac{\ell!^2}{(2\ell)!} \leq 2\sqrt{\ell} \leq 2\ell.$$
This implies that 
\be N(m,\ell) = \frac{\ell}{m(m-\ell+1)} \binom{m}{\ell}^2 
\geq \binom{2m}{2\ell} / 2m^2.
\label{eq:Nara-lb}\ee
Next, we observe that $\binom{2m}{2\ell+1} \leq \binom{2m}{2\ell} + \binom{2m}{2\ell +2}$, and so by comparing coefficients, we see that 
$$\frac{(1+\sqrt{x})^3}{x} \sum_{\ell=1}^m \binom{2m}{2\ell} \geq (1+\sqrt{x})^{2m}.$$
Combining this with \eq{Nara-lb} completes the proof of the Lemma.
\end{proof}

\section{Approach 1: Feynman diagrams}\label{sec:matt}
%\section{Random Product States}

\subsection{Reduction to the Gaussian ensemble}
\label{sec:gaussian}
We begin by showing how all the moments of the normalized ensemble are
always upper-bounded by the moments of the Gaussian ensemble.
A similar argument was made in
\cite[Appendix B]{BHLSW03}.  In both cases, the principle is that Gaussian
vectors can be thought of as normalized vectors together with some
small fluctuations in their overall norm, and that by convexity the
variability in norm can only increase the variance and other higher
moments.

\begin{lemma}\label{lem:normal-gaussian}
\bit\item[(a)] For all $p,d,k,m$ and all strings $\vs\in[p]^m$, 
\be e^{-\frac{m^2}{2d}}{\hat{E}_d[\vs]} \leq {E_d[\vs]}
\leq {\hat{E}_d[\vs]}.\label{eq:Ed-gauss-ineq}\ee
\item[(b)] For all $p,d,k,m$, 
\be e^{-\frac{m^2k}{2d}}{\hat{E}_{p,d,k}^m} \leq 
E_{p,d,k}^m\leq \hat{E}_{p,d,k}^m.\label{eq:Epdk-gauss-ineq}\ee
\eit
\end{lemma}

\begin{proof}
First note that
\be
E_d[\vs]
=\Bigl( \prod_{s=1}^p \int_{|\varphi_s|^2=1} {\rm d}\mu(\varphi_s) \Bigr)
\langle \varphi_{s_1},\varphi_{s_2} \rangle
\langle \varphi_{s_2},\varphi_{s_3} \rangle ...
\langle \varphi_{s_m},\varphi_{s_1} \rangle.
\ee
where the integral is over $\ket{\varphi_s}\in\bbC^d$ constrained to
$\braket{\varphi_s}{\varphi_s}=1$.

Next, 
for a given choice of $s_1,...,s_m$, let $\mu_s(s_1,...,s_m)$ denote the
number of times the letter $s$ appears.  For example, for
$s_1,...,s_m=1,2,2,1,3$ we have $\mu_1=2,\mu_2=2,\mu_3=1$.
Then,
let us introduce variables $r_s$ and use
$\int_0^\infty {\rm d}r_s \exp(-\D r_s^2/2)
r_s^{2\D+2\mu_s-1}
\frac{\D^{\D+\mu_s}}{(\D+\mu_s)!}=1$ to write
\begin{eqnarray}
E_d[s_1,s_2,...,s_m]
\label{fd}
&=& \Bigl( \prod_{s=1}^p \int_{|\varphi_s|^2=1} {\rm d}\mu(\varphi_s) 
\int_0^\infty {\rm d}r_s e^{-\frac{\D r_s^2}{2}}
r_s^{2\D-1}
\frac{\D^{\D+\mu_s}}{(\D+\mu_s)!}
\Bigr)
r_{s_1} \langle \varphi_{s_1},\varphi_{s_2} \rangle 
r_{s_2} \langle \varphi_{s_2},\varphi_{s_3} \rangle ...
r_{s_m} \langle \varphi_{s_m},\varphi_{s_1} \rangle \\ \nonumber
&=&
\Bigl( \prod_{s=1}^p 
\frac{\D!\D^{\mu_s}}{(\D+\mu_s)!}
(\frac{\D}{2\pi})^\D
\int {\rm d}\hat{\varphi}_s 
\exp(-\D |\hat{\varphi}_s|^2/2) \Bigr)
\langle \hat{\varphi}_{s_1},\hat{\varphi}_{s_2} \rangle
\langle \hat{\varphi}_{s_2},\hat{\varphi}_{s_3} \rangle ...
\langle \hat{\varphi}_{s_m},\hat{\varphi}_{s_1} \rangle,
\\ \nonumber & = &
\Bigl( \prod_{s=1}^p 
\frac{\D!\D^{\mu_s}}{(\D+\mu_s)!}\Bigr)
\hat{E}_d[\vs]
\end{eqnarray}
where the integral on the second line is over all
$\ket{\hat{\varphi}_s}\in\bbC^d$, 
with
\be
\ket{\hat{\varphi}_s} = r_s \ket{\varphi_s}.
\ee

Then, since the integral
\be
\hat{E}_d[\vs]=
\label{eq:gauss-int}
\Bigl( \prod_{s=1}^p 
\l(\frac{\D}{2\pi}\r)^\D
\int {\rm d}\hat{\varphi}_s 
\exp(-\D |\hat{\varphi}_s|^2/2) \Bigr)
\langle \hat{\varphi}_{s_1},\hat{\varphi}_{s_2} \rangle
\langle \hat{\varphi}_{s_2},\hat{\varphi}_{s_3} \rangle ...
\langle \hat{\varphi}_{s_m},\hat{\varphi}_{s_1} \rangle
\ee
is positive, and
\be
1\geq \prod_{s=1}^p \frac{\D!\D^{\mu_s}}{(\D+\mu_s)!}\geq 
\frac{1}{\l(1+\frac{1}{d}\r)\cdots\l(1+\frac{m}{d}\r)} \geq
e^{-\frac{m(m+1)}{2d}},
\label{eq:gaussian-LB}
\ee
we establish \eq{Ed-gauss-ineq}.

Since $E_{p,d,k}^m$ (resp. $\hat{E}_{p,d,k}^m$) is a sum over
$E_d[\vs]^m$ (resp. $\hat{E}_d[\vs]^m$), each of which is nonnegative,
we also obtain \eq{Epdk-gauss-ineq}.  This completes the proof of the
lemma.\end{proof}

From now on, we focus on this sum:
\be
\label{eq:gauss-int2}
\hat{E}_{p,d,k}^m = 
\sum_{s_1=1}^p \sum_{s_2=1}^p ... \sum_{s_m=1}^p
\Bigl[
\Bigl( \prod_{s=1}^p 
(\frac{\D}{2\pi})^\D
\int {\rm d}\hat{\varphi}_s 
\exp(-\D |\hat{\varphi}_s|^2/2) \Bigr)
\langle \hat{\varphi}_{s_1},\hat{\varphi}_{s_2} \rangle
\langle \hat{\varphi}_{s_2},\hat{\varphi}_{s_3} \rangle ...
\langle \hat{\varphi}_{s_m},\hat{\varphi}_{s_1} \rangle\Bigr]^k
\ee
We introduce a diagrammatic way of evaluating this sum.

\subsection{Diagrammatics}
This section now essentially follows standard techniques in field
theory and random matrix theory as used in \cite{FZ97}.  The main
changes are: first, for $k=1$, our diagrammatic notation will be the
same as the usual ``double-line'' notation, while for $k>1$ we have a
different notation with multiple lines.  Second, the recursion
relation (\ref{recrel}) is usually only evaluated at the fixed point
where it is referred to as the the ``Green's function in the
large-$d$ approximation'' (or more typically the large-$N$
approximation), 
while we study how the number of diagrams changes as the number of
iterations $a$ is increased in order to verify that the sum of
Eq.~(\ref{greensdef}) is convergent.  Third, we only have a finite
number, $2m$, of vertices, so we are able to control the corrections
which are higher order in $1/\D$ or $1/p$.  In contrast,
Ref. \cite{FZ97}, for example, considers Green's functions which are
sums of diagrams with an arbitrary numbers of vertices.

Integrating Eq.~(\ref{eq:gauss-int}) over $\hat{\varphi}_s$ generates
$\prod_s \mu_s!$ diagrams, as shown in Fig.~1.  Each diagram
is built by starting with one incoming directed line on the left and one outgoing line on
the right, with $m$ successive pairs of vertices as shown in
Fig.~1(a).  We then join the lines coming out of the vertices
vertically, joining outgoing lines with incoming lines, to make all
possible combinations such that, whenever a pair of lines are joined
between the $i$-th pair of vertices and the $j$-th pair of vertices,
we have $s_i=s_j$.  Finally, we join the rightmost outgoing line to the
leftmost incoming line; then the resulting diagram forms a number of closed loops.
The value of Eq.~(\ref{eq:gauss-int}) is equal to the sum over such
diagrams of $\D^{l-m}$, where $l$ is the number of closed loops in the
diagram.  Two example diagrams with closed loops are shown in Fig.~1(b).

\begin{figure}[tb]
\centerline{
\includegraphics[scale=0.4]{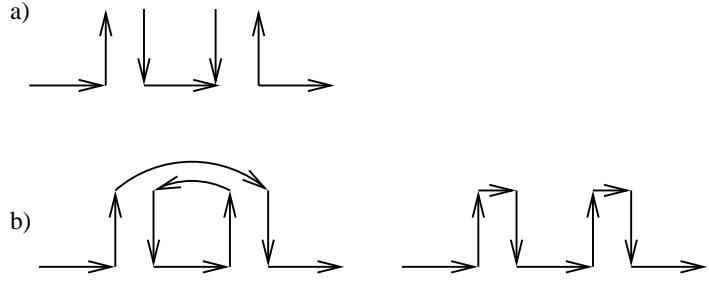}}
\caption{a) Vertices for $\hat{E}_d[s_1,s_2]$.
b) Example diagrams for $\hat{E}_d[s_1,s_2]$ with $s_1=s_2$.  The diagram on the
left has $l=m=2$ while the diagram on the right (which is also present for $s_1\neq s_2$)  has $l=1$.}
\end{figure}

\begin{figure}[tb]
\centerline{
\includegraphics[scale=0.4]{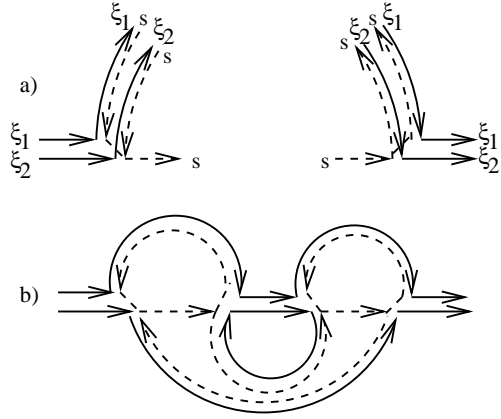}}
\caption{a) Vertices for diagrams.  The left vertex corresponds to
  $\hat\Phi^\dag$ and the right vertex corresponds to $\hat\Phi$.
b) An example diagram with $m=2$ and $k=2$.  There are $l_n=1+2=3$ 
loops of
solid lines and $l_p=1$ disconnected objects of dashed lines.
\label{fig:diagram}}
\end{figure}

Similarly, the sum of Eq.~(\ref{eq:gauss-int2}) can also be written
diagrammatically.  There are $k$ incoming lines on the left and $k$
outgoing lines on the right.  We have $m$ successive pairs of vertices
as shown in Fig.~\ref{fig:diagram}(a).  Each vertex has now $k$ {\it pairs} of lines
connected vertically: either the solid lines in the pairs are outgoing
and the dashed lines are incoming or vice-versa, depending on whether
the vertex has incoming solid lines on the horizontal or outgoing.  We
label the incoming solid lines by indices $\xi_1,...,\xi_k\in[d]$, which we
refer to as color indices, and then alternately assign to lines along
the horizontal axis either a single index of the form $s\in[p]$ for the
dashed lines, which we refer to as flavor indices, or $k$ different
color indices of the form $\xi_1,...,\xi_k\in[d]$ for the solid lines.  Each
of the $k$ lines in a set of $k$ parallel solid lines is also labelled
by a ``copy index'', with the top line labelled as copy $1$, the
second as copy $2$, and so on, up to copy $k$.

Each of the $k$ pairs of lines coming from a vertex is labelled with
a color index $\xi$ and a flavor index $s$, as well as a copy index.
The copy index on a vertical solid line is the same as the copy index
of the solid line it connects to on the horizontal, so a given
vertex has $k$ distinct copy indices, ranging from $1...k$.
Each diagram consists of a way of joining different pairs of
vertical lines, subject to the rule that when we join two
vertical lines, both have the same copy index;
thus, if a given vertical line comes from the $k'$-th row, $1\leq k' \leq k$,
then it must join to a line which also comes from the $k'$-th row.

The value of a diagram is equal to $\D^{-m}$ times the number of
possible assignments of values to the indices, such that whenever two
lines are joined they have the same indices.  The solid lines break up
into some number $l_n$ different closed loops; again, when counting
the number of closed loops, we join the solid lines leaving on the
right-hand side of the diagram to those entering on the left-hand side
of the diagram.  Since all solid lines in a loop have the same copy
index, we have $l_n=l_{n,1}+l_{n,2}+...+l_{n,k}$, where $l_{n,k'}$ is
the number of loops of solid lines with copy index $k'$.  The dashed
lines $s$ come together in vertices where $k$ different lines meet.
Let $l_p$ denote the number of different disconnected sets of dashed
lines.  Then, the value of a diagram is equal to \be \D^{-mk} \D^{l_n}
p^{l_p}.  \ee Note, we refer to disconnected sets of lines in the case
of dashed lines; this is because multiple lines meet at a single
vertex; for $k=1$ these sets just become loops.  An example diagram is
shown in Fig.~\ref{fig:diagram}(b) for $k=2$.

Let $c_m^k(l_n,l_p)$ equal the number of diagrams with given $l_n,l_p$ 
for
given $m,k$.
Then,
\be
\label{sumw}
\hat{E}_{p,d,k}^m
=\sum_{l_n\geq 1} \sum_{l_p\geq 1} c_m^k(l_n,l_p) 
\D^{-mk} \D^{l_n} p^{l_p}.
\ee

\subsection{Rainbow Diagrams}
An important set of diagrams are the so-called ``rainbow diagrams'', which
will be the dominant contributions to the sum (\ref{sumw}).
We define these rainbow diagrams with the following iterative
construction.

We define a group of $k$ solid lines or a single dashed line to
be an open rainbow diagram as shown in Fig.~\ref{fig:rainbow}(a).  
We also define any diagram which can
be contructed as in Fig.~\ref{fig:rainbow}(b,c)
to be a open rainbow diagram, where the $k$ solid lines or one dashed
line with a filled circle may be
replaced by any open rainbow diagram.  We say that the rainbow diagrams
in Fig.~\ref{fig:rainbow}(b,c) has solid and dashed external lines respectively.

In general all open rainbow diagrams can be constructed from the
iterative process described in Fig.~\ref{fig:rainbow}(b,c), with one ``iteration''
consisting of replacing one of the filled circles in Fig.~\ref{fig:rainbow}(b,c) with
one of the diagrams in Fig.~\ref{fig:rainbow}.  The diagrams in Fig.~\ref{fig:rainbow}(a) require zero
iterations, and each iteration adds one vertex.  For example, in
Fig.~4(a,b) we show the two diagrams with solid external lines which
require two iterations to construct for $k=1$.  We define a rainbow
diagram to be any open rainbow diagram where we assume that the right
outgoing edge and left incoming edge are solid lines and are
connected.

\begin{figure}[tb]
\centerline{
\includegraphics[scale=0.4]{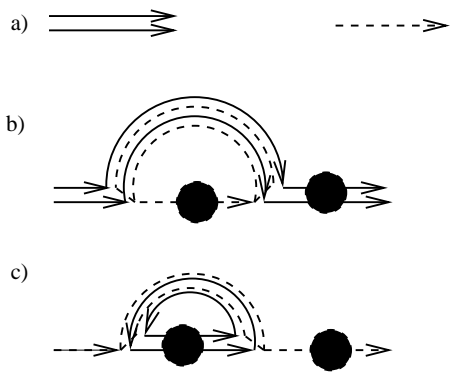}}
\caption{Iterative construction of rainbow diagrams for
$k=2$.  The solid lines with a filled circle denotes any open rainbow
diagram as does the dashed line with a filled circle.
\label{fig:rainbow}}
\end{figure}

\begin{figure}[tb]
\centerline{
\includegraphics[scale=0.4]{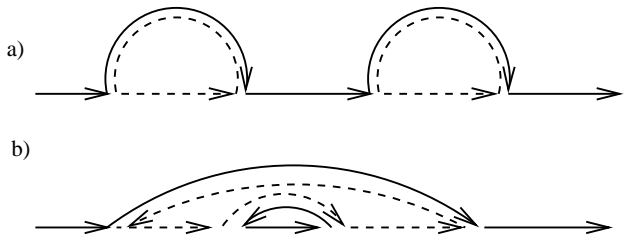}}
\caption{(a,b) Rainbow diagrams which require two iterations for $k=1$.}
\end{figure}

\subsection{Combinatorics of Diagrams and Number of Loops}
We now go through several claims about the various diagrams.  The goal 
will
be to count the number of diagrams for given $l_n,l_p$.
First, we claim that for the rainbow diagrams
\be
\label{saturate}
l_n+k l_p=(m+1)k,
\ee
as may be directly verified from the construction.
Next, we claim that for any diagram
\be
\label{sumb}
l_n+k l_p \leq (m+1)k.
\ee
From Eq.~(\ref{saturate}) the rainbow diagrams saturate this
bound (\ref{sumb}).  We claim that it suffices to show Eq.~(\ref{sumb})
for $k=1$ in order to show Eq.~(\ref{sumb}) for all $k$.  To see this,
consider any diagram for $k>1$.  Without loss
of generality, suppose $l_{n,1}\geq l_{n,k'}$ for all $1\leq k' \leq k$.
Then, $l_n+k l_p\leq k (l_{n,1}+p)$.
We then remove
all the solid lines on the horizontal with copy indices $2...k$, as
well as all pairs of lines coming from a vertex with copy indices
$2...k$.  Having done this, both the solid and the dashed lines form
closed loops, since only two dashed lines meet at each vertex.  The
new diagram is a diagram with $k=1$.  The
number of loops of solid lines is $l_{n,1}$, while the number of
loops of dashed lines in the new diagram,
$l_p'$, is greater than or equal to $l_p$ since we have removed dashed
lines from the diagram.  Thus, if we can show Eq.~(\ref{sumb}) for $k=1$,
it will follow that $l_{n,1}+l_p'\leq (m+1)$ and so $l_n+kl_p \leq (m+1)k$.

To show Eq.~(\ref{sumb}) for $k=1$, we take the given diagram, and
make the replacement as shown between the left and right half of
Fig.~5(a): first we straighten the diagram out as shown in the middle
of Fig.~5(a), then we replace the double line by a wavy line connected
the solid and dashed lines.  Finally, we take the point where the
solid line leaves the right-hand side of the diagram and connects to
the solid line entering the left-hand side and put a single dot on
this point for reference later as shown in Fig.~5(b,c).  Having done
this, the diagram consists of closed loops of solid or dashed lines,
with wavy lines that connect solid to dashed lines, and with one of
the closed loops of solid lines having a dot on it at one point.

This procedure gives an injective mapping from diagrams written as in
Fig.~2 to diagrams written as in Fig.~5.  However, this mapping is not
invertible; when we undo the procedure of Fig.~5(a), we find that some
diagrams can only be written as in Fig.~2 if there are two or more
horizontal lines.  The diagrams which are the result of applying this
procedure to a diagram as in Fig.~2 with only one horizontal line are
those that are referred to in field theory as contributions to the
``quenched average,'' while the sum of all diagrams, including those
not in the quenched average, is referred to as the ``annealed
average''.  To determine if a diagram is a contribution to the
quenched average, start at the dot and then follow the line in the
direction of the arrow, crossing along a wavy line every time it is
encountered, and continuing to follow solid and dashed lines in the
direction of the arrow, and continuing to cross every wavy line
encountered.  Then, a diagram is a contribution to the quenched
average if and only if following the lines in this manner causes one
to traverse the entire diagram before returning to the starting point,
while traversing wavy lines in both directions.  As an example,
consider the diagram of Fig.~5(c): this diagram is not a contribution
to the quenched average, as can be seen by traversing the diagram, or
by re-drawing the diagram as in Fig.~5(d) which requires two
horizontal solid lines\footnote{Such annealed diagrams are
  contributions to the average of the product of two (or more) traces
  of powers of $\hat M_{p,d,k}$.}.  If a diagram is a contribution to the
quenched average, then traversing the diagram in this order (following
solid, dashed, and wavy lines as above) corresponds to traversing the
diagram writen as in Fig.~2 from left to right.

Since all diagrams are positive, we can bound the
sum of diagrams which are contributions to the quenched average
by bounding the sum of all diagrams as in Fig.~5.
The number of wavy lines is equal to $m$.
The diagram is connected so therefore the number of solid plus dashed
loops, which is equal to $l_{n,1}+l_p'$, is at most
equal to the number of wavy lines plus one.  Therefore, Eq.~(\ref{sumb})
follows.  From this construction, 
the way to saturate Eq.~(\ref{sumb}) is to make a diagram which is a tree 
whose
nodes are closed loops of dashed and solid lines and whose edges are 
wavy
lines; that is, a diagram such that the removal of any wavy line breaks
the diagram into two disconnected pieces.  These trees are the
same as the rainbow diagrams above.
In Fig.~5(b) we show the two different trees which correspond
to the rainbow diagrams of Fig.~4.

\begin{figure}[tb]
\centerline{
\includegraphics[scale=0.4]{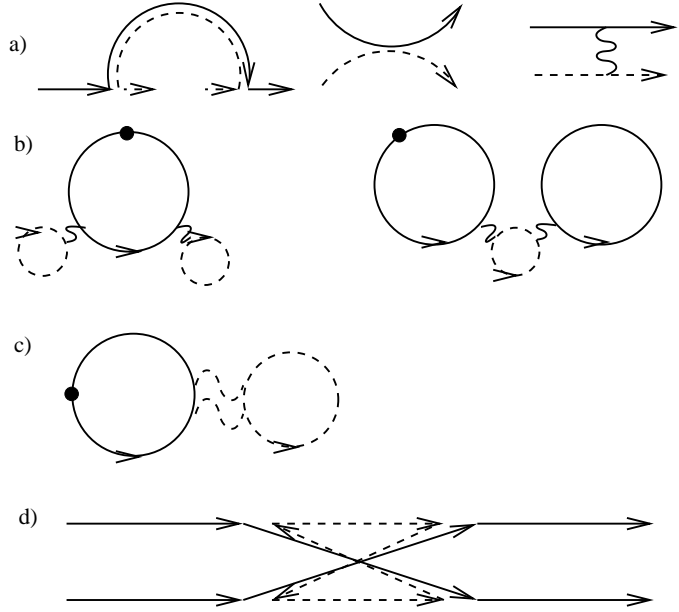}}
\caption{(a) Deformation of diagram.  (b)  Deformation of diagrams
in Fig.~4(a,b). (c) Example of a diagram which contributes to the annealed
average but not the quenched average. (d) Same diagram as in (c).}
\end{figure}

Next, we consider the diagrams which are not rainbow diagrams.
First, we consider the case $k=1$.  Let $d=m+1-l_n-l_p\geq 0$.  If
$d>0$, then the diagram is not a rainbow diagram.  However, if $d>0$, 
using
the construction above with closed loops connected by wavy lines, there
are only $l_n+l_p$ loops connected by more than $l_n+l_p-1$ wavy lines; 
this
implies that the diagram is not a tree (using the notation of
Fig.~5) or a rainbow diagram (using the notation of
Fig.~4), and hence it is possible to
remove $d$ different lines and arrive at a diagram which is a rainbow
diagram.  Thus, all diagrams with $2m$ vertices and
$d>0$ can be formed by taking rainbow diagrams with $2m-d$ vertices and
adding $d$ wavy lines; these wavy lines can be added in at most
$[m(m-1)]^{d}$ different ways.  Thus, for $k=1$ we have
\begin{eqnarray}
\label{ratiobnd1}
m+1-l_n-l_p=d>0 \, \rightarrow \,
c_m^1(l_n,l_p)&\leq & c_{m-d}^1(l_n,l_p) m^{2d}
\end{eqnarray}

We now consider the number of diagrams which are not rainbow diagrams
for $k>1$.  We consider all diagrams, including those
which contribute to the annealed average, but we write the diagrams as in
Fig.~2, possibly using multiple horizontal lines.  Consider
first a restricted class of diagrams: those diagrams for which, for
every vertex with $k$ pairs of lines leaving the vertex, all $k$ of those
pairs of lines connect with pairs of lines at the {\it same} vertex.  This
is not the case for, for example, the diagram of Fig.~\ref{fig:diagram}(b), where of the
two pairs of lines leaving the leftmost vertex, the top pair reconnects at
the second vertex from the left, while the bottom pair reconnects at the
rightmost vertex.  However, for a diagram in this restricted class, the
counting of diagrams is exactly the same as in the case $k=1$, since the
diagrams in this restricted class are in one-to-one correspondence with
those for $k=1$.  So, the number of diagrams in this restricted class,
$c_{m,r}^k$, obeys
\begin{eqnarray}
\label{ratiobnd}
(m+1)k-l_n-kl_p=d>0 \, \rightarrow \,
c_{m,r}^k(l_n,l_p)&\leq & c_{m-d,r}^k(l_n,l_p) [m(m-1)]^{d/k}
\end{eqnarray}

Now, we consider a diagram which is not in this restricted class.
Locate any vertex with incoming solid lines on the horizontal and
an outgoing dashed line on the horizontal, such that
not all pairs of lines leaving this vertex
reconnect at the same vertex.  Call this vertex
$v_1$.  Then, find any other vertex
to which a pair of lines leaving vertex $v_1$ reconnects.
Call this vertex $v_2$.  Let there be $l$ pairs of lines leaving vertex $v_1$
which do not conenct to $v_2$, and similarly $l$ pairs of lines entering 
$v_2$
which do not come from $v_1$, with $1\leq l \leq k-1$.  Label these
pairs of
lines $L^1_{1},L^1_2,...,L^1_k$ and $L^2_1,L^2_2,...,L^2_k$, respectively.
Let these lines connect to pairs of lines $M^1_1,M^1_2,...,M^1_k$ and
$M^2_1,M^2_2,...,M^2_k$ respectively.
Let $v_3$ be the vertex just to the right of $v_1$, so that the dashed
line entering $v_1$ comes from $v_3$, and similarly let $v_4$ be the 
vertex
just to the left of $v_2$, so that the dashed line leaving $v_2$ goes into
$v_4$, as shown in \fig{reconnect}(a).
Then, we determine
if there is a way to
re-connect pairs of lines
so that now $L^1_{l'}$ connects to $L^2_{l'}$
and $M^1_{l'}$ connects to $M^2_{l'}$ for all $l'$ in some subset of
$\{1,...,l\}$ such that the diagram splits into exactly two disconnected
pieces.  If there is, then we find the smallest subset of
$\{1,...,l\}$ with this property (making an arbitrary choice if there is
more than one such subset) and make those reconnections.  Let
${\cal V}_1,{\cal V}_2$ denote the two disconnected subsets of vertices
after the reconnections.  By making these reconnections, then, we
are reconnecting precisely the pairs of lines which originally connected
vertices in set ${\cal V}_1$ to those in ${\cal V}_2$; if there are
$l_c$ such lines, then we increase $l_n$ by $l_c\geq 1$.
Thus, we increase $l_p$ by one and
also increases $l_n$ by at least $1$.  We then modify the diagram to
rejoin the two pieces: the dashed line leaving to the right of
vertex $v_1$ connects to it some vertex $w_1$ in the same piece, and 
there
is some other dashed line in the other piece which connects two vertices
$v_1',w_1'$; we re-connect these dashed lines so that $v_1$ connects to
$w_1'$ and $v_1'$ connects to $w_1$.  This reduces $l_p$ by $1$ back to
its original value and
makes the diagram connected.
Thus, we succeed in increasing $l_n+kl_p-mk$ by at least $1$.

\begin{figure}[tb]
\centerline{
\includegraphics[scale=0.4]{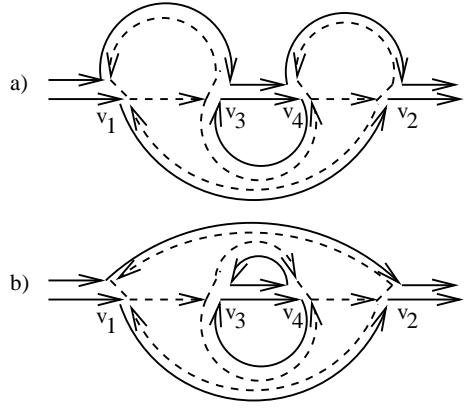}}
\caption{(a) Diagram of Fig.~\ref{fig:diagram}(b) with vertices
  $v_1,v_2,v_3,v_4$ marked, 
for a particular choice of $v_1,v_4$.  (b) Result of applying re-connection
procedure to diagram.
\label{fig:reconnect}}
\end{figure}

On the other hand, if no such subset exists, we re-connect all pairs
of lines for all $1\leq l' \leq l$, as shown in Fig.~\ref{fig:reconnect}(b).  The resulting
diagram must be connected (if not, then there would have been a subset
of lines which could be re-connected to split the line into exactly two
disconnected pieces).  Then,
there are two cases: the first case is when the dashed line leaving
$v_2$ does not connect to $v_1$ (so that $v_2\neq v_3$ and $v_1\neq
v_4$) {\it and} it is possible to re-connect the dashed lines joining
$v_1$ to $v_3$ and $v_2$ to $v_4$ so that now $v_2$ is connected to
$v_1$ and $v_3$ is connected to $v_4$ without breaking the diagram
into two disconnected pieces.  In this first case, we then also make
this re-connection of dashed lines, which increases $l_p$ by one,
while keeping the diagram connected.  However, in this case, the
initial re-connection of pairs of lines may have reduced $l_n$ by at
most $l$.  Thus, in this case $k l_p+l_n-mk$ is increased by at least
$1$.  The second case is when either $v_2$ connects to $v_1$ already
or it is not possible to make the re-connection of dashed lines
without splitting the diagram into two pieces.  This is the case in
Fig.~\ref{fig:reconnect}(b).  In this case, however, 
$l_n$ 
must have increased 
by at least $1$ by
the initial re-connection of pairs of lines\footnote{To see why
  $l_n$ must have been increased by at least one when reconnecting
  pairs of lines, in the case where making the reconnection of the
  dashed line would split the diagram into two disconnected pieces,
  let ${\cal V}_1,{\cal V}_2$ denote the vertices in these two
  disconnected pieces.  Then, by reconnecting the pairs of lines,
  there are no longer any solid lines joining ${\cal V}_1$ to ${\cal
    V}_2$, so $l_n$ increases by $l'\geq 1$.}  and thus again we
increase $k l_p+l_n-mk$ by at least $1$.

Repeating this procedure, we ultimately arrive
at a diagram in the restricted class above.
At each step, 
we succeed in reducing $d=(m+1)k-l_n-k l_p$ by
at least unity, either 
by increasing $l_n$ by at least
$1$ and $l_p$ by $2$, or by increasing $l_p$ by $1$ and reducing $l_n$ 
by
at most $k-1$.
Given a diagram in the restricted class, we can further reduce $d$ following
Eq.~(\ref{ratiobnd}).
Then, any diagram can
be found by starting with a diagram in the restricted class and
undoing this procedure; at each step in undoing the procedure
we have at most $m^2 (m-1)^{2(k-1)}$ choices (there are at most
$m^2$ choices for $v_1,v_2$,
and then we must re-connect at most $2(k-1)$ pairs of lines).
Thus,
for $(m+1)k-l_n-k l_p=d>0$
we have
\begin{eqnarray}
\label{ratiobnd2}
c_m^k(l_n,l_p) & \leq  &
m^{2k} c_{m}^k(l_n+1,l_p) \\ \nonumber 
&&+m^{2k}
c_m^k(l_n-k+1,l_p+1) \\ \nonumber
&&+m^2
c_{m-1}^k(l_n,l_p).
\end{eqnarray}
This implies that for $d>0$,
\begin{eqnarray}
\label{ratiobnd3}
\sum_{l_n,l_p}^{l_n+l_p=m+1-d} c_m^k(l_n,l_p) d^{l_n} p^{l_p}
& \leq  &
\delta \sum_{l_n,l_p}^{l_n+l_p=m+1-(d-1)} c_m^k(l_n,l_p) d^{l_n} p^{l_p},
\end{eqnarray}
with
\be
\delta=
\frac{m^{2k}}{\D}
+\frac{m^{2k}\D^{k-1}}{p}
+\frac{m^2}{\D^k}=
(1+x^{-1})\frac{m^{2k}}{\D}
+\frac{m^2}{\D^k}.
\ee

\subsection{Bound on Number of Rainbow Diagrams}
Finally, we provide a bound on the number of rainbow diagrams.  Let us
define $S_v^a(l_n,l_p)$ to equal to the number of open rainbow
diagrams with solid lines at the end, with $v$ vertices, $l_n$ loops
of solid lines (not counting the loop that would be formed by
connected the open ends), and $l_p$ disconnected sets of dashed lines,
which may be constructed by at most $a$ iterations of the process
shown in Fig.~\ref{fig:rainbow}.  Similarly, define $D_v^a(l_n,l_p)$ to equal to the
number of open rainbow diagrams with dashed lines at the end, with $v$
vertices, $l_n$ loops of solid lines (not counting the loop that would
be formed by connected the open ends), and $l_p$ disconnected sets of
dashed lines, which may be constructed by at most $a$ iterations of
the process shown in Fig.~\ref{fig:rainbow}.  These open rainbow diagrams obey
$l_n/k+l_p=m$.  Define the generating function\footnote{The limit as
  $a\rightarrow \infty$ of this generating functional is equal to, up
  to a factor $1/z$ in front, the Green's function in the large-$\D$
  limit usually defined in field theory.}
%; see the ``Heuristics'' section for more details.}
\begin{eqnarray}
\label{greensdef}
G_s^{(a)}(z,\D,p)&=&
\sum_v \sum_{l_n} \sum_{l_p} z^{-v/2} \D^{-vk/2} \D^{l_n} p^{l_p} 
S_v^a(l_n,l_p),
\\ \nonumber
G_d^{(a)}(z,\D,p)&=&
\sum_v \sum_{l_n} \sum_{l_p} z^{-v/2} \D^{-vk/2} \D^{l_n} p^{l_p} 
D_v^a(l_n,l_p).
\end{eqnarray}
Then, we have the recursion relations, which come from Fig.~\ref{fig:rainbow}(b,c):
\begin{eqnarray}
\label{recrel}
G_s^{(a)}(z,\D,p)=1+z^{-1} x G_d^{(a-1)}(z,\D,p) G_s^{(a-1)}(z,\D,p), \\ 
\nonumber
G_d^{(a)}(z,\D,p)=1+z^{-1} G_s^{(a-1)}(z,\D,p) G_d^{(a-1)}(z,\D,p).
\end{eqnarray}
%where \be x=p/\D^k.\ee  
First, consider the case
$x\leq 1$.
From Eq.~(\ref{recrel}), $G_d^{(a)}(z,\D,p)=1+(G_s^{(a)}(z,\D,p)-1)/x$ for
all $a$, so that we have the recursion
$G_s^{(a)}(z,\D,p)=1+z^{-1} x G_s^{(a-1)}(z,\D,p) (1+(G_s^{(a-1)}(z,\D,p)-1)/
x)=
1+z^{-1}(x-1)G_s^{(a-1)}(z,\D,p)+z^{-1} G_s^{(a-1)}(z,\D,p)^2$.
The fixed points of this recursion relation are given by
\be
G_s(z,\D,p)\equiv
\frac{z^{-1}(1-x)+1\pm \sqrt{(z^{-1}(1-x)+1)^2-4z^{-1}}}{2z^{-1}}.
\label{eq:rainbow-fp}
\ee
Define
\be
z_0=\l(\frac{1+x-2\sqrt{x}}{(1-x)^2}\r)^{-1}=
(1+\sqrt{x})^2.
\ee
Then, for $z>z_0$, Eq.~(\ref{recrel}) has two real fixed points, while
at $z=z_0$, Eq.~(\ref{recrel}) has a single fixed point at
\be
G_s(z_0,\D,p)=
\frac{z_0}{2}\l[1+z_0^{-1}(1-x)\r] = 1+ \sqrt{x} = \sqrt{z_0} >1.
\ee
Since $G_s^{(0)}(z,\D,p)=G_d^{(0)}(z,\D,p)=1$ which is smaller than the
fixed point, we find that $G_s^{(a)}(z,\D,p)$ increases
montonically with $a$ and remains bounded above by
$G_s(z,\D,p)$.  All rainbow diagrams with $2m$ vertices can be found
after a finite number (at most $m$) iterations of Fig.~\ref{fig:rainbow}(b,c) so
\be
\label{rainbowsum}
\sum_{l_n,l_p}^{l_n+l_p=m+1} c_m^k(l_n,l_p) 
\D^{-mk} \D^{l_n} p^{l_p}\leq
p z_0^{m} 
G_s(z_0,\D,p).
\ee

Alternately, if $x\geq 1$, we use
$G_s^{(a)}(z,\D,p)=1+(G_d^{(a)}(z,\D,p)-1)x$, to 
get the recursion
$G_d^{(a)}(z,\D,p)=1+z^{-1} G_d^{(a)}(z,\D,p) (1+(G_d^{(a)}(z,\D,p)-1)x)=
1+z^{-1}(1-x)G_d^{(a-1)}(z,\D,p)+z^{-1}x G_d^{(a-1)}(z,\D,p)^2$.
Then, again for $z=z_0$ this recursion has a single fixed point
and $G_s^{(a)}(z,\D,p)$ increases monotonically with $a$ and remains
bounded by $G_s(z_0,\D,p)$.

\subsection{Sum of All Diagrams}
We now bound the sum of all diagrams (\ref{sumw}) using the bound
on the sum of rainbow diagrams (\ref{rainbowsum}) and Eq.~(\ref{ratiobnd3}).
\begin{eqnarray}
\label{allbnd}
\sum_{j\geq 0}
\sum_{l_n,l_p}^{l_n+l_p=m+1-j} c_m^k(l_n,l_p) 
\D^{-mk} \D^{l_n} p^{l_p}
& \leq &
p z_0^{m} 
G_s(z_0,\D,p)
\sum_{j\geq 0}
\delta^j.
\end{eqnarray}
Then, if $\delta<1$ we have
\be
\label{finalbnd}
\hat{E}_{p,d,k}^m
\leq
 \frac{p}{1-\delta} z_0^{m} 
{G}_s(z_0,\D,p) = 
 \frac{p}{1-\delta} z_0^{m+\frac{1}{2}} 
\ee
We can pick
$m$ of order $d^{1/2k}$ and still have $\delta \leq 1/2$.  Then we can use $\bbE[\|\hat{M}_{p,d,k}\|] \leq (\hat{E}_{p,d,k}^m)^{1/m}$ to bound
\ba
\bbE[\|\hat{M}_{p,d,k}\|] &\leq 
(1+\sqrt{x})^2
\cdot \exp\l(\frac{\ln(2p\sqrt{z_0})}{m}\r)
\nn
& =
(1+\sqrt{x})^2+
O\l(\frac{k \ln(\D)}{\D^{\frac{1}{2k}}}\r),
\label{eq:diag-eig-bound}
\ea
as claimed in \cor{eig-ub}. 
We are assuming in the $O()$ notation in this bound that $x=\Theta(1)$.
%If we want a bound that works for even very small values of  $x$ then we can take $m \sim (\sqrt{d} \ln(d))^{\frac{1}{k+1}}$ to obtain a correction of size $O(\ln^{\frac{k}{k+1}}(d)d^{-\frac{1}{2(k+1)}})$.

\section{Approach 2: combinatorics and representation theory}
\label{sec:aram}

This section gives a second proof of \thm{trace} that uses facts about
symmetric subspaces along with elementary combinatorics.  The
fundamentals of the proof resemble those of the last section in many
ways, which we will discuss at the end of this section.  However, the
route taken is quite different, and this approach also suggests
different possible extensions.

Recall that we would like to estimate
$$E_{p,d,k}^m = \sum_{\vs\in[p]^m} E_d[\vs]^k.$$
Our strategy will be to repeatedly reduce the string $\vs$ to simpler 
forms.  Below we will describe two simple methods
for reducing $\vs$ into a possibly shorter string $R(\vs)$ such that
$E_d[\vs]$ equals $E_d[R(\vs)]$, up to a possible multiplicative factor
of $1/d$ to some power.  Next we will consider two important special
cases.  First are the {\em completely reducible} strings: $\vs$ for
which the reduced string $R(\vs)$ is the empty string.  These are
analogous to the rainbow diagrams in \secref{matt} and their
contribution can be calculated exactly (in
\secref{complete-reduce}).  The second special case is 
when $\vs$ is {\em irreducible}, meaning that $R(\vs)=\vs$; that is,
neither simplification steps can be applied to $\vs$.  These strings
are harder to analyze, but fortunately make a smaller contribution to
the final sum.
In \secref{irred-strings}, we use representation theory to give upper
bounds for $E_d[\vs]$ for irreducible strings $\vs$, and thereby to
bound the overall contribution from irreducible strings.  Finally, we
can describe a general string as an irreducible string punctuated with
some number of 
repeated letters (defined below) and completely reducible strings.
The overall sum can then be bounded using a number of methods; we will
choose to use a generating function
approach, but inductively verifying the final answer would also be
straightforward.

{\em Reducing the string:}  Recall that $E_d[\vs]=\tr
\varphi_{s_1}\cdots \varphi_{s_m}$, where each $\ket{\varphi_s}$ is a
unit vector randomly chosen from $\bbC^\D$.  We will use the following
two reductions to simplify $\vs$.
\benum
\item {\em Remove repeats.} Since $\varphi_a$ is a pure state,
  $\varphi_a^2=\varphi_a$ and we can replace every instance of {\em aa}
  with {\em a} in $\vs$ without changing $E_d[\vs]$.  Repeatedly
  applying this means that if $s_i=s_{i+1}=\cdots=s_j$, then
  $E_d[\vs]$ is unchanged by deleting positions $i+1,\ldots,j$.
  Here we identify position $i$ with $m+i$ for all $i$, so that
  repeats can wrap around the end of the string: e.g. the string
  11332221 would become 321.  

\item {\em Remove unique letters.}  Since $\bbE[\varphi_a]=I/\D$ for
  any $a$, we can replace any letters which appear only once with
  $I/\D$.  
%We say that these letters are ``unique letters.''
  Thus, if $s_i\neq s_j$ for all $j\neq i$ then $E_d[\vs]=E_d[\vs']/\D$,
  where $\vs'\in [p]^{m-1}$ is obtained from $\vs$ by deleting the
  position $i$.  Repeating this process results in a string where
  every letter appears at least twice and with a multiplicative factor
  of $1/\D$ for each letter that has been removed.  Sometimes the
  resulting string will be empty, in which case we say
  $E_d[\emptyset]=\tr I = \D$.  Thus for strings of length one,
  $E_d[a] = E_d[\emptyset]/\D = \D/\D = 1$.
\eenum

We will repeatedly apply these two simplification steps until no
further simplifications are possible.  Let $R(\vs)$ denote the
resulting (possibly empty) string.  Recall from above that
when $R(\vs)=\emptyset$, we say $\vs$ is completely reducible, and
when $R(\vs)=\vs$, we say $\vs$ is irreducible.  The sums over these
two special cases are described by the following two Lemmas.

\begin{lemma}\label{lem:complete-reduce}
\be
\frac{1}{d^k}
\sum_{\substack{\vs\in[p]^m \\ R(\vs)=\emptyset}} E_d[\vs]^k
= \sum_{\ell=1}^m N(m,\ell)\frac{(p)_\ell}{d^{kl}}
\leq \beta_m\l(\frac{p}{d^k}\r)
\leq %\l(1+\sqrt{\frac{p}{d^k}}\r)^{2m} = 
\lambda_+^{m}
.\label{eq:Nara-sum}
\ee

\end{lemma}

We will prove this Lemma and discuss its significance in
\secref{complete-reduce}.  It will turn out that the completely
reducible strings make up the dominant contribution to $E_{p,d,k}^m$
when $m$ is not too large.  Since \eq{Nara-sum} is nearly 
independent of $k$ (once we fix $x$ and $p$), this means that 
$E_{p,d,k}^m$ is also nearly independent of $k$.  It remains only to
show that the sub-leading order terms do not grow too quickly with
$k$.  Note that this Lemma establishes the lower bound of \thm{trace}.

For the irreducible strings we are no longer able to give an exact
expression.  However, when $m$ is sufficiently small relative to $\D$
and $p$, we have the following nearly tight bounds.
\begin{lemma}\label{lem:irred-strings}
If $m<\min(d^{k/6}/2^{1+k/2},(p/5000)^{\frac{1}{2k+12}})$ then 
\be
\sum_{\substack{\vs\in[p]^m \\ R(\vs)=\vs}} E_d[\vs]^k
\leq
\frac{e^{\frac{km}{2d^{1/3}}}}
{\l(1 - \frac{5000m^{2k+12}}{p}\r)
\l(1-\frac{2^{2+k}m^2}{d^{\frac{k}{3}}}\r)}
x^{\frac{m}{2}}
\label{eq:irred-sum}\ee
Additionally, when $m$ is even, the left-hand side of \eq{irred-sum}
is $\geq x^{\frac{m}{2}} e^{-\frac{m^2}{2p}}$.
\end{lemma}
The proof is in \secref{irred-strings}.  Observe that when
$m\in o(d^{k/6})\cap o(p^{\frac{1}{2k+12}})$ and $m$ is even, we bound the 
sum 
on the LHS of \eq{irred-sum} by
$(1\pm o(1))x^{m/2}$.  We also mention that there is no factor of $1/d^k$ on the LHS, so that when $x=O(1)$ and $m$ satisfies the above condition, the contribution from irreducible strings is a $O(1/d^k)$ fraction of the contribution from completely reducible strings.

Next, we combine the above two Lemmas to bound all strings that are not covered by \lemref{complete-reduce}.
\begin{lemma}\label{lem:mixed-strings}
If $m<\min(d^{k/6}/2^{1+k/2},(p/5000)^{\frac{1}{2k+12}})$ then 
\be \sum_{\substack{\vs\in[p]^m \\ R(\vs)\neq \emptyset}} E_d[\vs]^k
\leq 
\frac{e^{\frac{km}{2d^{1/3}}}}
{\l(1 - \frac{5000m^{2k+12}}{p}\r)
\l(1-\frac{2^{2+k}m^2}{d^{\frac{k}{3}}}\r)}
m \lambda_+^{m+\frac{1}{2}}.
\label{eq:mixed-sum}\ee
\end{lemma}
The proof is in \secref{all-strings}.

To simplify the prefactor in \eq{mixed-sum}, we assume that 
 $m<\min(2d^{1/3}/k, d^{k/6}/2^{2+k/2},(p/5000)^{\frac{1}{2k+12}}/2)$, so that the RHS of \eq{mixed-sum} becomes simply $\leq 12 m\lambda_+^{m+\frac{1}{2}}$.  By \lemref{beta-comp}, this is $\leq \frac{24m^3\lambda_+^2}{x} \beta_m(x)      $.
    Then we combine \lemref{complete-reduce} and \lemref{mixed-strings} to obtain the bound
\be e_{p,d,k}^m \leq 
\l(1 + \frac{24 m^3\lambda_+^2 }{p}\r)\beta_m(x)
\ee
which is a variant of the upper-bound in \thm{trace}.  It is tighter than \eq{main-e-bound}, but holds for a more restricted set of $m$.  If we express the upper bound in terms of $\lambda_+^m$ then we can skip \lemref{complete-reduce} and obtain simply
\be
e_{p,d,k}^m \leq \l(1 + \frac{12 m \sqrt{\lambda_+}}{d^k}\r) \lambda_+^m.
\label{eq:tighter-ub}\ee

\subsection{Completely reducible strings}
\label{sec:complete-reduce}

We begin by reviewing some facts about Narayana numbers
from \cite{Nara1, Nara2}.
The Narayana number
\be N(m,\ell)=\frac{1}{m}\binom{m}{\ell-1}\binom{m}{\ell} =
\frac{1}{\ell}\binom{m}{\ell-1}\binom{m-1}{\ell-1}
\label{eq:Nara-formula}\ee
counts the number of valid bracketings of $m$ pairs of parentheses in
which the sequence {\em ()} appears $\ell$ times.  A straightforward
combinatorial proof of \eq{Nara-formula} is in
\cite{Nara2}.  When we sum \eq{Nara-formula} over
$\ell$ (e.g. if we set $x=1$ in \eq{Nara-sum}) then we obtain the
familiar Catalan numbers $\frac{1}{m+1}\binom{2m}{m}$.

% Like the Catalan numbers, the Narayana numbers can also be
% interpreted in terms of lattice paths from $(0,0)$ to $(m,m)$ where
% every point in the path satisfies $i\leq j$ and where there are
% $\ell$ steps, each moving from $(i,j)$ to $(i',j')$ with $i'-i$ and
% $j'-j$ both positive integers.  The Catalan numbers count the number
% of such paths without the restriction on the number of steps.

We can now prove \lemref{complete-reduce}.  The combinatorial
techniques behind the Lemma have been observed before\cite{Nara1,
  Nara2}, and have been applied to the Wishart distribution in
\cite{another-moment, forrester}.

\mn{\em Proof:}
For a string $\vs$ such that $R(\vs)=\emptyset$, let $\ell$ be the
number of distinct letters in $\vs$.  In the process of reducing $\vs$
to the empty string we will ultimately remove $\ell$ unique letters,
so that $E_d[\vs]^k=d^{k(1-\ell)}$.  It remains now only to count the
number of different $\vs$ that sastify
$R(\vs)=\emptyset$ and have $\ell$ distinct letters.

Suppose the distinct letters in $\vs$ are $S_1, S_2, \ldots, S_\ell
\in [p]$.  We order them so that the first occurrence of $S_i$ is
earlier than the first occurrence of $S_{i+1}$ for each $i$.  Let
$\vsi$ be the string obtained from $\vs$ by replacing each instance of
$S_i$ with $i$.  Then $\vsi$ has the first occurrences of
$1,2,\ldots,\ell$ appearing in increasing order and still satisfies
$R(\vsi)=\emptyset$ and $E_d[\vsi]^k=d^{k(1-\ell)}$.  Also, for each
$\vsi$, there are $p!/(p-\ell)! 
\leq p^\ell$ corresponding $\vs$.

It remains only to count the number of distinct $\vsi$ for a given
choice of $m$ and $\ell$.  We claim that this number is given by
$N(m,\ell)$.  Given $\vsi$, define $a_i$ to be the location of the
first occurrence of the letter $i$ for $i=1,\ldots,\ell$.  Observe
that 
\be 1 = a_1 < a_2 < \cdots < a_\ell \leq m.
\label{eq:a-constraint}\ee  Next, define $\mu_i$
to be the total number of occurrences of $i$ in $\vsi$, and define
$b_i = \sum_{j=1}^i \mu_j$ for $i=1,\ldots,\ell$.  Then 
\be 1\leq b_1 < b_2 < \cdots < b_\ell = m 
\label{eq:b-constraint}\ee
Finally, we have 
\be a_i \leq b_i \qquad\text{for each } i=1,\ldots,\ell.
\label{eq:ab-constraint}\ee

Ref.~\cite{Nara1} proved that the number of
$(a_1,b_1),\ldots,(a_\ell,b_\ell)$ satisfying \eq{a-constraint},
\eq{b-constraint} and \eq{ab-constraint} is $N(m,\ell)$.  Thus, we
need only prove that $\vsi$ is uniquely determined by
$(a_1,b_1),\ldots,(a_\ell,b_\ell)$.  The algorithm for finding $\vsi$
is as follows.
\begin{tabbing}
For $t=1,\ldots,m$. \= \\
\> If $t=a_i$ then set $s:=i$.\\
\> Set $\sigma_t := s$. \\
\> Set $\mu_s := \mu_s - 1$.\\
\> While ($\mu_s = 0$) set $s := s-1$.\\
\end{tabbing}
In other words, we start by placing 1's until we reach $a_2$.  Then we
start placing 2's until we've either placed $\mu_2$ 2's, in which case
we go back to placing 1's; or we've
reached $a_3$, in which we case we start placing 3's.  The general
rule is that we keep placing the same letter until we either encounter
the next $a_i$ or we run out of the letter we were using, in which
case we go back to the last letter we placed.

To show that $\vsi$ couldn't be constructed in any other way, first
note that we have $\sigma_{a_i}=i$ for each $i$ by definition.  Now
fix an $i$ and 
examine the interval between $a_i$ and $a_{i+1}$.  Since it is before
$a_{i+1}$, it must contain only letters in $\{1,\ldots,i\}$.  Using
the fact that $R(\vsi)=\emptyset$, we know that $\vsi$ cannot contain
the subsequence $j$-$i$-$j$-$i$ (i.e. cannot be of the form $\cdots j
\cdots i \cdots j \cdots i$).  
We now consider two cases.  

Case (1) is
that $\mu_i \geq a_{i+1}-a_i$.  In this case we must have $\sigma_t=i$
whenever $a_i<t<a_{i+1}$.  Otherwise, this would mean that some
$s\in\{1,\ldots,i-1\}$ appears in this interval, and since $s$ must
have appeared earlier as well ($s<i$ so $a_s<a_i$ and
$\sigma_{a_s}=s$), then no $i$'s can appear later in the string.
However, this contradicts the fact that $\mu_i \geq a_{i+1}-a_i$.
Thus if $\mu_i \geq a_{i+1}-a_i$ then the entire interval between
$a_i$ and $a_{i+1}$ must contain $i$'s.

Case (2) is that $\mu_i < a_{i+1}-a_i$.  This means that there exists
$t$ with $a_i<t<a_{i+1}$ and $\sigma_t\in\{1,\ldots,i-1\}$; if there is
more than one then take $t$ to be the lowest (i.e. earliest).  Note
that $\sigma_{t'}\neq i$ for all $t'>t$; otherwise we would have a
$\sigma_t$-$i$-$\sigma_t$-$i$ subsequence. Also, by definition
$\sigma_{t'}=i$ for $a_i\leq t' < t$.  Since this is the only place
where $i$ appears in the string, we must have $t=a_i + \mu_i$.  Once
we have placed all of the $i$'s, we can proceed inductively to fill
the rest of the interval with letters from $\{1,\ldots,i-1\}$.

In both cases, $\vsi$ is uniquely determined by $a_1,\ldots,a_\ell$
and $b_1,\ldots,b_\ell$ (or equivalently, $\mu_1,\ldots,\mu_\ell$).
This completes the proof of the equality in \eq{Nara-sum}.
\qed

Before continuing, we will mention some facts about Narayana numbers
that will later be useful.
Like the Catalan numbers, the Narayana
numbers have a simple generating function; however, since they have
two parameters the generating function has two variables.  If we
define
\be F(x,y) = \sum_{0\leq \ell\leq m< \infty} N(m,\ell) x^\ell y^m,
\label{eq:Nara-gf-def}\ee
then one can show\cite{Nara1,Nara2} (but note that \cite{Nara2} takes the sum over $m
\geq 1$) that
\be F(x,y) = \frac{ 1 + (1-x)y - \sqrt{1- 2(1+x)y + (1-x)^2y^2}}{2y}
\label{eq:Nara-gf}.\ee
We include a proof for convenience.  First, by convention $N(0,0)=1$.
Next, an arrangement of $m$ pairs of parentheses can start either with
{\em ()} or {\em ((}.  Starting with {\em ()} leaves $N(m-1,\ell-1)$ ways to
complete the string.  If the string starts with {\em ((} then suppose the {\em )} paired with 
the first {\em (} is the $i^{\text{th}}$ {\em )} in the string.  We
know that 
$2\leq i\leq m$ and that the first $2i$ characters must contain
exactly $i$
{\em (}'s and $i$ {\em )}'s.  Additionally, the $2i-1^{\text{st}}$ and $2i^{\text{th}}$ characters 
must both be {\em )}'s.
Let $j$ be the number of appearances of {\em ()} amongst these first $2i$ characters.  
Note that $j\leq \min(i-1,\ell)$, and that {\em ()} appears $\ell-i$ times in the last $2m-2i$ 
characters.  Thus there are 
$$\sum_{i=2}^m \sum_{j=1}^{\min(i-1,\ell)} N(i-1,j) N(m-i, \ell-j) = 
-N(m-1,\ell) +\sum_{i=1}^m \sum_{j=1}^{\min(i-1,\ell)} N(i-1,j) N(m-i, \ell-j)
$$
ways to complete a string starting with {\em ((}.
Together, these imply that
\be N(m,\ell) = N(m-1,\ell-1) - N(m-1,\ell) + 
\sum_{i=1}^m \sum_{j=1}^{\min(i-1,\ell)} N(i-1,j) N(m-i, \ell-j)
\label{eq:Nara-recur},\ee
which we can state equivalently as an identify for the generating
function \eq{Nara-gf-def}:
\be F = 1 + xy F + y (F^2 - F),
\label{eq:Nara-gf-recur}\ee
which has the solution \eq{Nara-gf}.  (The sign in front of the square
root can be  
established from $1 = N(0,0) = F(x,0)$.) 

{\em Connection to \secref{matt}:}
Observe that \eq{Nara-gf} matches \eq{rainbow-fp} once we make the
substitution $y=z^{-1}$.   Indeed it can be shown that rainbow diagrams have a one-to-one correspondence with valid arrangements of parentheses, and thus can be enumerated by the 
Narayana numbers in the same way.

{\em Connection to free probability:} Another set counted by the Narayana numbers is the   
set of noncrossing partitions of $[m]$ into $\ell$ parts.  The
non-crossing condition means that we never have $a<b<c<d$ with $a,c$
in one part of the partition and $b,d$ in another; it is directly
analogous to the property that $\vsi$ contains no subsequence of the
form $j$-$i$-$j$-$i$.

To appreciate the significance of this, we return to the classical
problem of throwing $p$ balls into $d$ bins.  The occupancy of a
single bin is $z = z_1+\ldots+z_p$ where $z_1,\ldots,z_p$ are
i.i.d. and have $\Pr[z_i=0] = 1-1/d$, $\Pr[z_i=1]=1/d$. 
One can readily verify that
$$\bbE[z^m] = \sum_{\ell=1}^m |\Par(m,\ell)| 
\frac{(p)_\ell}{d^\ell},$$
where $\Par(m,\ell)$ is the set of (unrestricted) partitions of $m$
into $\ell$ parts.

This is an example of a more general phenomenon in which convolution
of classical random variables involves partitions the same way that
convolution of free random variables involves non-crossing
partitions.  See Ref.~\cite{Speicher-NCPart} for more details.

\subsection{Irreducible strings}
\label{sec:irred-strings}

As with the completely reducible strings, we will break up the sum
based on the powers of $p$ and $d$ which appear.  However, while in
the last section $p$ and $1/d^k$ both depended on the single parameter
$\ell$, here we will find that some terms are smaller by powers of
$1/p$ and/or $1/d$.  Our strategy will be to identify three
parameters---$\ell$, $c_2$, and $\hat{\mu}_2$, all defined below---for
which the leading contribution occurs when all three equal $m/2$.  We
show that this contribution is proportional to $\sqrt{x}^m$ and that
all other values of $\ell$, $c_2$, and $\hat{\mu}_2$ make negligible
contributions whenever $m$ is sufficiently small.

Again, we will let $\ell$ denote the number of unique letters in
$\vs$.   We will also let $S_1,\ldots,S_\ell\in [p]$ denote these unique
letters.  However, we choose them so that $S_1<S_2<\cdots<S_\ell$,
which can be done in
\be\binom{p}{\ell}\leq \frac{p^\ell}{\ell!}
\label{eq:irred-p-factor}\ee
ways.  Again, we let $\vsi\in[\ell]^m$ be the string that results from
replacing all the instances of $S_i$ in $\vs$ with $i$.  However,
because of our different choice of $S_1,\ldots,S_\ell$, we no longer
guarantee anything about the ordering of $1,\ldots,\ell$ in $\vsi$.

We will also take $\mu_a$ be the frequency of $a$ in $\vsi$ for each
$a=1,\ldots,\ell$.  We also define $\hat{\mu}_b$ to be the number of
$a$ such that $\mu_a=b$.  Observe that 
\ba \ell &= \sum_b \hat{\mu}_b 
\label{eq:l-partition} \\
 m &= \sum_{a=1}^\ell \mu_a = \sum_b b \hat{\mu}_b.
\label{eq:m-partition}\ea
Also recall that since $R(\vsi)=\vsi$, $\vsi$ has no
repeats or unique letters.  Thus $\mu_a\geq 2$ for each $a$, or
equivalently $\hat{\mu}_1=0$.  This also implies that $\ell\leq m/2$.
Since \eq{irred-p-factor} is maximised when $\ell= \lfloor
\frac{m}{2}\rfloor$, we will focus on this case first and then show that
other values of $\ell$ have smaller contributions.  Moreover
\eq{m-partition} implies that
$\hat{\mu}_2\leq \ell$ and \eq{l-partition}, \eq{m-partition} and the
fact that $\hat{\mu}_1=0$ imply that $\hat{\mu}_2\geq
3\ell-m$.  Together we have
\be 3\ell-m \leq \hat{\mu}_2 \leq \ell.
\label{eq:mu2-bounds}\ee
  Thus  $\ell$ is close to $m/2$ if and only if $\hat{\mu}_2$
is as well.  This will be useful because strings will be easier to
analyze when almost all letters occur exactly twice.

We now turn to the estimation of $E_d[\vsi]$.
To analyze $E_d[\vsi] = \bbE [\tr
\varphi_{\sigma_1}\varphi_{\sigma_2}\cdots\varphi_{\sigma_m}]$, we 
first introduce the cyclic shift operator
$$C_m = \sum_{i_1,\ldots,i_m\in[\D]}
\ket{i_1,\ldots,i_m}\bra{i_2,\ldots,i_m,i_1}.$$
Then we use the identity
\be \tr [\varphi_{\sigma_1}\varphi_{\sigma_2}\cdots\varphi_{\sigma_m}] = 
\tr[ C_m (\varphi_{\sigma_1}\otimes \varphi_{\sigma_2} \ot
\cdots\otimes \varphi_{\sigma_m})].
\label{eq:cyclic-trace}\ee
Next, we take the expectation.  It is a well-known consequence of
Schur-Weyl duality (see e.g. Lemma 1.7 of \cite{matthias}) that
\be \bbE[\varphi^{\otimes t}] = \frac{\sum_{\pi\in\cS_t}
  \pi}{d(d+1)\cdots(d+t-1)}.
\label{eq:Schur-average}\ee
We will apply this to \eq{cyclic-trace} by inserting
\eq{Schur-average} in the appropriate locations as given by $\vsi$.
Let $\cS_\vsi:=\{\pi\in \cS_m : \sigma_i = \sigma_{\pi(i)} \forall i\in [m]\}$ be
the set of permutations that leaves $\vsi$ (or equivalently $\vs$)
invariant.  Then $|\cS_\vsi| = \mu_1!\cdots\mu_\ell!$ and
\begin{subequations}\label{eq:Eds-perm}
\ba
E_d[\vsi] &= \bbE[\tr C_m(\varphi_{\sigma_1}\otimes \varphi_{\sigma_2} \ot
\cdots\otimes \varphi_{\sigma_m})] \\
& = \tr C_m \frac{\sum_{\pi\in\cS_\vsi}\pi}{\prod_{i=1}^\ell
  d(d+1)\cdots(d+\mu_i-1)}\\
& \leq \tr C_m \frac{\sum_{\pi\in\cS_\vsi} \pi}{\prod_{i=1}^\ell d^{\mu_i}}\\
& = \frac{\sum_{\pi\in\cS_\vsi}\tr C_m \pi}{d^m}\\
& = \sum_{\pi\in\cS_\vsi} d^{\cyc(C_m\pi)-m}.
\ea \end{subequations}
This last equality follows from the fact that for any permutation
$\nu$ acting on $(\bbC^d)^{\ot m}$, we have that 
\be \tr \nu = d^{\cyc(\nu)},\label{eq:perm-trace}\ee
 where $\cyc(\nu)$ is the number of cycles of $\nu$.
(Eq.~(\ref{eq:perm-trace}) can be proven by first considering the case when
$\cyc(\nu)=1$ 
and then decomposing a general permutation into a tensor product of
cyclic permutations.)

To study $\cyc(C_m\pi)$, we introduce a graphical notation for
strings. For any string $\vsi$, define the {\em letter graph} $G$ to
be a directed graph with $\ell$ vertices such that for
$i=1,\ldots,\ell$, vertex $i$ has in-degree and out-degree both equal
to $\mu_i$.  (For brevity, we will simply say that $i$ has degree
$\mu_i$.) Thus there are a total of $m$ edges.  The edges leaving and
entering vertex $i$ will also be ordered.  To construct the edges in
$G$, we add an edge from $s_i$ to $s_{i+1}$ for $i=1,\ldots,m$, with
$s_{m+1}:=s_1$.  The ordering on these edges is given by the order we
add them in.  That is, if letter $a$ appears in positions
$i_1,i_2,\ldots$ with $i_1<i_2<\cdots$, then the first edge out of $a$
is directed at $s_{i_1+1}$, the second out-edge points at $s_{i_2+1}$,
and so on.  Likewise, $a$'s incoming edges (in order) come from
$s_{i_1-1}, s_{i_2-1}, \ldots$.

Now we think of the incoming and outgoing edges of a
vertex as linked, so that if we enter on the $j^{\text{th}}$ incoming
edge of a vertex, we also exit on the $j^{\text{th}}$ outgoing edge.
This immediately specifies a cycle through some or all of $G$.  If we
use the ordering specified in the last paragraph then the cycle is in
fact an Eulerian cycle (i.e. visits each edge exactly once) that
visits the vertices in the order $s_1,s_2,\ldots,s_m$.  Thus, from a
letter graph $G$ and a starting vertex we can 
reconstruct the string $\vsi$ that was used to generate $G$.

\begin{figure}
\begin{center}
\includegraphics{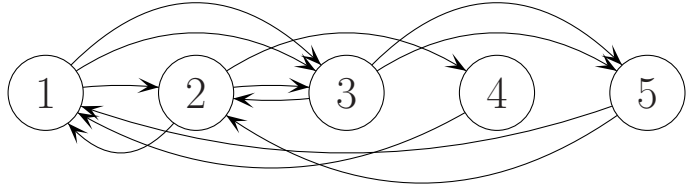}
\caption{The letter graph corresponding to the string 123241351352.
\label{fig:letter-graph}}
\end{center}
\end{figure}

The letter graph of $\vsi$ can also be used to give a cycle
decomposition of $C_m\pi$.  Any permutation $\pi\in\cS_{\vsi}$ can be
thought of as permuting the mapping between in-edges and out-edges for
each vertex.  The resulting number of edge-disjoint cycles is exactly
$\cyc(C_m\pi)$.  To see this, observe that $\pi$ maps $i_1$ to some
$i_2$ for which $\sigma_{i_1}=\sigma_{i_2}$ and then $C_m$ maps $i_2$
to $i_2+1$.  In $G$ these two steps simply correspond to following one of 
the
edges out of $i_1$.  Following the path (or the permutation) until it
repeats itself, we see that cycles in $G$ are equivalent to cycles in
$C_m\pi$.

We now use letter graphs to estimate \eq{Eds-perm}.
While methods for exactly enumerating cycle decompositions of
directed graphs do exist\cite{interlace},  for our purposes a crude upper bound will suffice.
Observe that because $\vsi$ contains no repeats, $G$ contains no
1-cycles.  Thus, the shortest cycles in $G$ (or equivalently, in
$C_m\pi$) have length 2.  Let $c_2(\pi)$ denote the number of 2-cycles
in $C_m\pi$ and $c_2^{\max} = \max_{\pi\in\cS_\vsi}c_2(\pi)$.
Sometimes we simply write $c_2$ instead of $c_2(\pi)$ when
 the argument is understood from context.
We now
observe that $c_2$ obeys bounds analogous to those in
\eq{mu2-bounds}.  In particular, $c_2^{\max} \leq \frac{m}{2}$, and for
any $\pi$,
\be \cyc(C_m\pi) \leq c_2(\pi) + \frac{m-2c_2(\pi)}{3} = 
\frac{m + c_2(\pi)}{3}
\label{eq:c2-bounds}.\ee

Since $c_2(\pi)\leq c_2^{\max} \leq m/2$, \eq{c2-bounds} implies that
$\cyc(C_m\pi)\leq m/2$.  Thus the smallest power of $1/d$ possible in
\eq{Eds-perm} is $\frac{m}{2}$.  When we combine this with
\eq{mu2-bounds}, we see that the leading-order contribution (in terms
of $p$ and $d$) is $O(x^{m/2})$, and that other terms are smaller by
powers of $1/p$ and/or $1/d$.  Additionally, this leading-order
contribution will have a particularly simple combinatorial factor.

{\em The leading-order term.}  Consider the case when $m$ is even and
$\ell=\hat{\mu}_2=c_2^{\max} = \frac{m}{2}$.  This corresponds to a graph
with $\ell$ vertices, each with in-degree and out-degree two.
Additionally, there is an ordering of the edges which organizes them into
$\ell$ 2-cycles.  Thus every vertex participates in exactly two
2-cycles.  Since the graph is connected, it must take the form of a
single doubly-linked loop.  Thus the letter graph of the leading-order
term is essentially unique.  See \fig{loop-graph} for an example when
$m=10$.  The only freedom here is the ordering of the vertices, which
can be performed in $\ell!$ ways.  Together with \eq{irred-p-factor},
this means the combinatorial contribution is simply
$\ell!\binom{p}{\ell}\leq p^\ell$. 

\begin{figure}
\begin{center}
\includegraphics{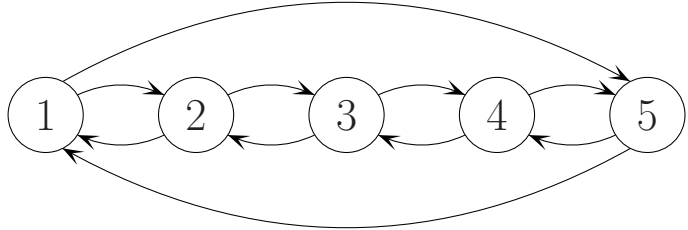}
\caption{Example of the letter graph for the case when $m=10$ and
$\ell=\hat{\mu}_2=c_2^{\max} = \frac{m}{2}$. The corresponding string is
1234543215. 
\label{fig:loop-graph}}
\end{center}
\end{figure}

Now we examine the sum in \eq{Eds-perm}.  Assume without loss
of generality that the
vertices $1,\ldots,\ell$ are connected in the cycle
$1-2-3-\cdots-\ell-1$.  Each vertex has two different configurations
corresponding to the two permutations in $\cS_2$.  In terms of the
letter graph these correspond to the two different ways that the two
incoming edges can be connected to the two outgoing edges.  Since
vertex $i$ has one edge both to and from each of $i\pm 1$, we can
either
\bit
\item[(a)] connect the incoming $i-1$ edge to the outgoing $i-1$ edge, and
  the  incoming $i+1$ edge to the outgoing $i+1$ edge (the closed
  configuration) ; or,
\item[(b)]  connect the incoming $i-1$ edge to the outgoing $i+1$ edge, and
  the  incoming $i+1$ edge to the outgoing $i-1$ edge (the open
  configuration).
\eit
These possibilities are depicted in \fig{rewiring}.

\begin{figure}
\begin{center}
\subfigure[closed configuration]
{\includegraphics[scale=0.8]{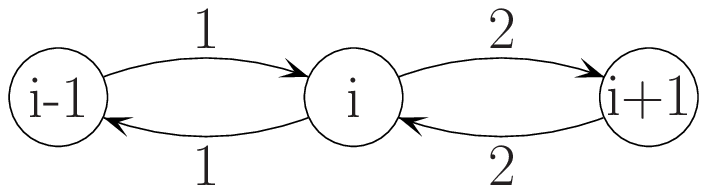}}
\hspace{2cm}
\subfigure[open configuration]
{\includegraphics[scale=0.8]{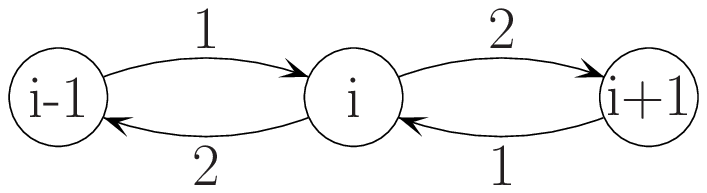}}
\end{center}
\caption{Vertex $i$ is connected to $i\pm 1$ by one edge in either
  direction.  These edges can be connected to each other in two ways,
  which are depicted in (a) and (b).  We call (a) a ``closed''
  configuration and (b) an ``open'' configuration.
\label{fig:rewiring}}
\end{figure}

Let $c$ denote the number of vertices in closed configurations.  These 
vertices
can be selected in $\binom{\ell}{c}$ ways.  If
$1\leq c\leq \ell$ then $c$ is also the number of cycles: to see this,
note that each closed configuration caps two cycles and each cycle
consists of a chain of open configurations that is capped by two closed
configurations on either end.  The exception is when $c=0$.  In this
case, there are two cycles, each passing through each vertex exactly
once.  Thus, the RHS of 
\eq{Eds-perm} evaluates (exactly) to
$$d^{2-m} + \sum_{c=1}^\ell \binom{\ell}{c} d^{c-m} =
d^{-\frac{m}{2}} \l[\l(1+\frac{1}{d}\r)^{\frac{m}{2}} +
d^{-\frac{m}{2}}(d^2 - 1)\r].$$
Combining everything, we find a contribution of
$x^{\frac{m}{2}}(1+o(1))$ as $d\ra\infty$.  In particular, when $m$ is
even this yields the lower bound claimed in \lemref{irred-strings}.
We now turn to the case 
when $c_2^{\max}$, $\ell$ and $\hat{\mu}_2$ are not all equal to $m/2$.

{\em The sum over all terms.}  Our method for handling arbitrary
values of $c_2^{\max}$, $\ell$ and $\hat{\mu}_2$ is to compare their
contribution with the leading-order term.  We find that if one of these
variables is decreased we gain combinatorial factors, but also need to
multiply by a power of $1/p$ or $1/d$.  The combinatorial factors will
turn out to be polynomial in $m$, so if $m$ is sufficiently small the
 contributions will be upper-bounded by a geometrically
decreasing series.  This process resembles (in spirit, if not in details) the process leading
to Eq.~(\ref{ratiobnd2}) in \secref{matt}.

Our strategy is to decompose the graph into a ``standard'' component
which resembles the leading-order terms and a ``non-standard''
component that can be organized arbitrarily.  The standard component
is defined to be the set of 2-cycles between degree-2 vertices.  When
$\ell=\hat{\mu}_2=c_2^{\max} = \frac{m}{2}$ the entire graph is in the
standard component, so when $\ell,\hat{\mu}_2,c_2^{\max} \approx
\frac{m}{2}$, the non-standard component should be small.  Thus, in
what follows, it will be helpful to keep in mind that the largest
contributions come from when
$\frac{m}{2}-\ell, \frac{m}{2}-\hat{\mu}_2, \frac{m}{2}-\cm$ are all
small, and so our analysis will focus on this case.

Begin by observing that there are $\ell-\hat{\mu}_2$ vertices with
degree greater than two.  Together these vertices have
$m-2\hat{\mu}_2$ in- and out-edges.  Thus, they (possibly together
with some of the degree-2 vertices) can participate in at most
%$\frac{m}{2} - \hat{\mu}_2$
$m-2\hat{\mu}_2$ 2-cycles.  Fix a permutation $\pi$ for which
$c_2(\pi) = c_2^{\max}$.  To account for all the 2-cycles, there must
be at least $c_2^{\max} 
- (m-2\hat{\mu}_2)$ 2-cycles between degree-2 vertices.  These
2-cycles amongst degree-2 vertices (the standard component) account
for $\geq 2c_2^{\max} - 2m + 4\hat{\mu_2}$ edges.  Thus the number of
non-standard edges entering and leaving the degree-2 vertices is $\leq
2\hat{\mu}_2 - (2c_2^{\max} - 2m + 4\hat{\mu_2}) = 2m-
2c_2^{\max}-2\hat{\mu}_2.$ Together we have $\leq 3m - 2c_2^{\max} -
4\hat{\mu_2}$ non-standard edges.

We now bound the number of ways to place the $m$ edges in $G$.  First,
we can order the degree-2 vertices in $\hat{\mu}_2!$ ways.  This
ordering will later be used to place the 2-cycles of the standard
component.  Next, we fix an arbitrary ordering for the
$\ell-\hat{\mu}_2$ vertices with degree larger than two.
% We also need to choose their degrees, which can be chosen to be any
% list of $\ell-\hat{\mu}_2$ integers, each larger than two, which sum
% to $m-2\hat{\mu}_2$.  Assuming $\hat{\mu}_2<\ell$, there are
% $\binom{m-2\ell+\hat{\mu}_2-1}{\ell-\hat{\mu}_2-1}$ ways of choosing
% the degree sequence $\mu_1,\ldots,\mu_\ell$.
We then place 
$$e_{\text{NS}} := 3m-2c_2^{\max} - 4\hat{\mu}_2$$
non-standard edges.  This can be done in $\leq m^{e_{\text{NS}}}$
ways.  One way to see this is that each non-standard edge has $m$
choices of destination, since we allow them to target specific
incoming edges of their destination vertex.  Call these destination
edges $\{I_1,\ldots,I_{e_{\text{NS}}}\}$.  These incoming edges
correspond to $e_{\text{NS}}$ outgoing edges, which we call
$\{O_1,\ldots,O_{e_{\text{NS}}}\}$, and which become the starting
points of the non-standard edges.  Without loss of generality we can
sort $\{O_1,\ldots,O_{e_{\text{NS}}}\}$ according to some canonical
ordering; let $\{O_1',\ldots,O_{e_{\text{NS}}}'\}$ be the sorted
version of the list.  Then we let $O_i'$ connect to $I_i$ for
$i=1,\ldots,e_{\text{NS}}$.  Since our ordering of
$\{I_1,\ldots,I_{e_{\text{NS}}}\}$ was arbitrary, this is enough to
specify any valid placement of the edges.  Additionally, our choices
of $\{I_1,\ldots,I_{e_\text{NS}}\}$ also determine the degrees
$\mu_1,\ldots,\mu_\ell$ since they account for all of the incoming
edges of the non-degree-2 vertices and out-degree equals in-degree.
  Note that nothing prevents
non-standard edges from being used to create 2-cycles between degree-2
vertices.  However we conservatively still consider such cycles to be
part of the non-standard component.

The remaining $m-e_{\text{NS}} = 2c_2^{\max} + 4\hat{\mu}_2 - 2m$ edges
(if this number is positive)
make up 2-cycles between degree-2 vertices, i.e. the standard
component.  Here we use the ordering of the degree-2 vertices.  After
the non-standard edges are placed, some degree-2 vertices will have
all of their edges filled, some will have one incoming and one
outgoing edge filled, and some will have none of their edges filled.
Our method of placing 2-cycles is simply to place them between
all pairs of neighbors (relative to our chosen ordering) whenever this
is possible.

We conclude that the total number of graphs is
\be \leq \hat{\mu}_2! m^{e_{\text{NS}}}
\leq\hat{\mu}_2! m^{3m-2c_2^{\max} - 4\hat{\mu}_2}
\leq\ell! m^{3m-2c_2^{\max} - 4\hat{\mu}_2}.
\label{eq:graph-count}\ee
In order to specify a string $\vsi$, we need to additionally choose a
starting edge.  However, if we start within the standard component,
the fact that we have already ordered the degree-2 vertices means that
this choice is already accounted for.  Thus we need only consider
\be e_{\text{NS}}+1 \leq 2^{e_{\text{NS}}} =
2^{3m-2c_2^{\max} - 4\hat{\mu}_2}
\label{eq:starting-positions}\ee
 initial edges, where we have used the fact that $1+a\leq 2^a$ for 
any integer $a$.
The total number of strings corresponding to given values of $\ell,
\hat{\mu}_2, \cm$ is then upper-bounded by the product of
\eq{starting-positions} and \eq{graph-count}:
\be \ell! (2m)^{3m-2c_2^{\max} - 4\hat{\mu}_2}.
\label{eq:string-count}\ee
Observe that this matches the combinatorial factor for the
leading-order term ($\cm=\hat{\mu}_2=\ell=m/2$) and then degrades
smoothly as $\cm,\hat{\mu}_2,\ell$ move away from $m/2$.

Finally, we need to evaluate the sum over permutations in
\eq{Eds-perm}.  Our choices for non-standard vertices are
substantially more complicated than the open or closed options we had
for the leading-order case.  Fortunately, it suffices to analyze only
whether each 2-cycle is present or absent.  Since a 2-cycle consists
of a pair of edges of the form $(i,j)$ and $(j,i)$, each such cycle
can independently be present or absent. Thus, while there are
$\mu_1!\cdots\mu_\ell!$ total elements of $\cS_\vsi$, we can break
the sum into $2^\cm$ different groups of
$(\mu_1!\cdots\mu_\ell!)/2^\cm$ permutations, each corresponding to a
different subset of present 2-cycles.  In other words, there are
exactly
$$\binom{\cm}{c} \frac{\mu_1!\cdots\mu_\ell!}{2^\cm}$$
choices of $\pi\in\cS_\vsi$ such that $c_2(\pi)=c$.  Using the fact
that $\cyc(C_m\pi)\leq (m+c_2(\pi))/3$, we have
$$E_d[\vsi] \leq \sum_{c=0}^{\cm}\binom{\cm}{c}
\frac{\mu_1!\cdots\mu_\ell!}{2^\cm}d^{\frac{m+c}{3}-m}
= \frac{\mu_1!\cdots\mu_\ell!}{2^\cm}d^{\frac{-2m+\cm}{3}}
\l(1 + d^{-\frac{1}{3}}\r)^{\cm}
$$
Finally, observe that $\mu_1!\cdots\mu_\ell!$ is a convex function of
$\mu_1,\ldots,\mu_\ell$ and thus is maximized when $\mu_1=m-2\ell+2$
and $\mu_2=\cdots=\mu_\ell=2$ (ignoring the fact that we have already
fixed $\hat{\mu}_2$).   Thus
\ba E_d[\vsi] &\leq (m - 2\ell+2)! 2^{\ell-1-\cm}
d^{\frac{-2m+\cm}{3}}\l(1+d^{-\frac{1}{3}}\r)^{\cm}\\
&\leq m^{m-2\ell}2^{\frac{m}{2}-\cm} %2^{\ell-\cm}
d^{\frac{-2m+\cm}{3}}e^{\frac{m}{2d^{1/3}}},
%\l(1+d^{-\frac{1}{3}}\r)^{\frac{m}{2}},
\label{eq:E0-bound}\ea
where in the last step we used the facts that $2\leq\ell\leq m/2$ and
$\cm\leq m/2$.

We now combine \eq{E0-bound} with the combinatorial factor in
\eq{string-count} to obtain
\ba \sum_{\substack{\vs\in[p]^m \\ R(\vs)=\vs}} E_d[\vs]^k 
&\leq
\sum_{0\leq \cm \leq \frac{m}{2}}
\sum_{0\leq \ell \leq \frac{m}{2}}
\sum_{3\ell-m\leq \hat{\mu}_2 \leq \ell}
\frac{p^\ell}{\ell!}
\ell! (2m)^{3m-2c_2^{\max} - 4\hat{\mu}_2}
\l[m^{m-2\ell}2^{\frac{m}{2}-\cm}
d^{\frac{-2m+\cm}{3}}e^{\frac{m}{2d^{1/3}}}\r]^k
\\&= x^{\frac{m}{2}}e^{\frac{km}{2d^{1/3}}}
\sum_{\substack{0\leq \cm \leq \frac{m}{2} \\
{0\leq \ell \leq \frac{m}{2}}\\
{3\ell-m\leq \hat{\mu}_2 \leq \ell}}}
\frac{p^{\ell-\frac{m}{2}}}{d^{\frac{k}{3}\l(\frac{m}{2}-\cm\r)}}
(2m)^{(m-2c_2^{\max})+(2m - 4\hat{\mu}_2)}
m^{k(m-2\ell)}2^{k(\frac{m}{2}-\cm)}
\label{eq:irred-sum-almost-done}\ea
We can bound the sum over $\hat{\mu_2}$ by introducing
$\alpha=\ell-\hat{\mu}_2$, so that
\be\sum_{3\ell-m\leq \hat{\mu}_2 \leq \ell}
(2m)^{2(m - 2\hat{\mu}_2)} = 
(2m)^{2m-4\ell}\sum_{\alpha=0}^{m-2\ell} \l(2m\r)^{4\alpha}
=(2m)^{2m-4\ell}(1 + 16m^4)^{m-2\ell}
\leq \l(65m^6\r)^{m-2\ell}
\label{eq:mu2-sum}
\ee
Substituting \eq{mu2-sum} in \eq{irred-sum-almost-done} and
rearranging, we obtain  
\ba
\sum_{\substack{\vs\in[p]^m \\ R(\vs)=\vs}} E_d[\vs]^k 
 & \leq 
x^{\frac{m}{2}}e^{\frac{km}{2d^{1/3}}}
\sum_{0\leq \cm \leq \frac{m}{2}}
\sum_{0\leq \ell \leq \frac{m}{2}} 
\l(\frac{5000m^{2k+12}}{p}\r)^{\frac{m}{2}-\ell}
\l(\frac{2^{2+k}m^2}{d^{\frac{k}{3}}}\r)^{\frac{m}{2}-\cm}
\\ & \leq
\frac{e^{\frac{km}{2d^{1/3}}}}
{\l(1 - \frac{5000m^{2k+12}}{p}\r)
\l(1-\frac{2^{2+k}m^2}{d^{\frac{k}{3}}}\r)}
x^{\frac{m}{2}}
.\ea
In the last step we have assumed that both terms in the denominator
are positive.  This completes the proof of
\lemref{irred-strings}. \qed

\subsection{Bounding the sum of all strings}
\label{sec:all-strings}

For any string $\vs\in[p]^m$ we will repeatedly remove
repeats and unique letters until the remaining string is irreducible.
Each letter in the original string either (a) appears in the final
irreducible string, (b) is removed as a repeat of one of the letters
appearing in the final irreducible string, or (c) is removed as part
of a completely reducible substring.  Call the letters A, B or C
accordingly.   Assign a weight of $\sqrt{x}^ty^t$ to each run of $t$ A's, a
weight of $y^t$ to each run of $t$ B's and of $\sum_{t=0}^\infty 
\sum_{\ell=0}^{t}N(t,\ell)x^\ell y^t$ to each run of $t$ C's. Here $y$ is an
indeterminant, but we will see below that it can also be thought of as
a small number.  We will define $G(x,y)$ to be the sum over all finite strings  of
A's, B's and C's, weighted according to the above scheme.  Note that
$[y^m]G(x,y)$ (i.e. the coefficient of $y^m$ in $G(x,y)$) is the
contribution from strings of length $m$.  

We now relate $G(x,y)$ to the sum in \eq{mixed-sum}.
Define
$$A_0 = 
\frac{e^{\frac{km}{2d^{1/3}}}}
{\l(1 - \frac{5000m^{2k+12}}{p}\r)
\l(1-\frac{2^{2+k}m^2}{d^{\frac{k}{3}}}\r)}
$$
so that \lemref{irred-strings} implies that the contribution from all
irreducible strings of length $t$ is $\leq A_0\sqrt{x}^t$ as long as $1\leq t\leq m$.  We will treat the $t=0$ case separately in \lemref{complete-reduce}, but for simplicity allow it to contribute a $A_0\sqrt{x}^0$ term to the present sum.  Similarly, we ignore the fact that there are no irreducible strings of length 1, 2, 3 or 5, since we are only concerned with establishing an upper bound here.
Thus
\be \sum_{\substack{\vs\in[p]^m \\ R(\vs)\neq \emptyset}} E_d[\vs]^k
\leq  A_0 [y^m]G(x,y) \leq A_0 \frac{G(x,y_0)}{y_0^m},
\label{eq:mixed-sum-gf-bound}\ee
where the second inequality holds for any $y_0$ within the radius of convergence of $G$.    We will choose $y_0$ below, but first give a derivation of $G(x,y)$.

To properly count the contributions from completely reducible substrings (a.k.a.~C's), we recall that $F(x,y)$ counts all C strings of length $\geq 0$.  Thus, it will be convenient to model a general string as starting with a run of 0 or more C's, followed by one or more steps, each of which places either an A or a B, and then a run of 0 or more C's.  (We omit the case where the string consists entirely of C's, since this corresponds to completely reducible strings.)  Thus, 
\be G(x,y) = F(x,y) \cdot \sum_{n\geq 1} \l[ y(1+\sqrt{x})F(x,y)\r]^n 
= \frac{y(1+\sqrt{x})F^2(x,y)}{1-y(1+\sqrt{x})F(x,y)},\ee
which converges whenever $F$ converges and $y(1+\sqrt{x})F<1$.  However, since we are only interested in the coefficient of $y^m$ we can simplify our calculations by summing over only $n\leq m$.  We also omit the $n=0$ term, which corresponds to the case of completely reducible strings, which we treat separately.  Thus, we have
$$G_m(x,y) := F(x,y) \cdot \sum_{n=1}^m  \l[ y(1+\sqrt{x})F(x,y)\r]^n, $$
and $G_m(x,y)$ satisfies $[y^m] G_m(x,y) = [y^m]G(x,y)$.

Now define $y_0 = \lambda_+^{-1} = (1+\sqrt{x})^{-2}$.
Rewriting $F$ as $\frac{1}{2}\l(y^{-1}+1-x-\sqrt{(y^{-1}-(1+x))^2-4x}\r),$  we see that $F(x,y_0) = 1 + \sqrt{x}$.  Thus $y_0(1+\sqrt{x})F(x,y_0) = 1$ and $G_m(x,y_0) = m (1+\sqrt{x})$.

Substituting into \eq{mixed-sum-gf-bound} completes the proof of the Lemma.

\subsection{Alternate models}
We now use the formalism from \secref{irred-strings} to analyze some
closely related random matrix ensembles that have been suggested by
the information locking proposals of \cite{LW-locking}.  The first ensemble we
consider is one in which each $\ket{\varphi_s^j}$ is a random unit
vector in $A_j \ot B_j$, then the $B_j$ system is traced out.  Let
$d_A = \dim A_1 = \ldots = \dim A_k$ and $d_B = \dim B_1 = \ldots =
\dim B_k$.  The resulting matrix is 
$$M_{p,d_A[d_B],k} := \sum_{\vs\in [p]^m} \bigotimes_{j=1}^k \tr_{B_j}
\varphi_{s_j}^j.$$ 
If $d_B \ll d_A$ then we expect the states $\tr_B \varphi_s$ to be
nearly proportional to mutually orthogonal rank-$d_B$ projectors and
so we expect $M_{p,d_A[d_B],k}$ to be nearly isospectral to
$M_{p,d_A/d_B,k} \ot \tau_{d_B}^{\ot k}$, where $\tau_d := I_d / d$.
Indeed, if we define $E_{p,d_A[d_B],k}^m := \tr M_{p,d_A[d_B],k}^m$
then we have  
\begin{lemma}
$$E_{p,d_A[d_B],k}^m  \leq E_{p,d_A/d_B,k}^m e^{\frac{m(m+1)kd_B}{2d_A}} d_B^{k(1-m)}.$$
\end{lemma}

\begin{proof}
Define $E_{d_A[d_B]}[\vs] = \tr (\tr_{B_1}(\varphi_{s_1}^1) \cdots \tr_{B_m}(\varphi_{s_m}^1)$.
Following the steps of \eq{Eds-perm}, we see that
\ba E_{d_A[d_B]}[\vs] 
&= \tr (C_m^{A^m} \ot I^{B^m}) \bbE (\varphi_{s_1} \ot \cdots \varphi_{s_m}) \\
& \leq \tr (C_m^{A^m} \ot I^{B^m}) \frac{\sum_{\pi\in\cS_{\vs}} \pi^{A^m} \ot \pi^{B_m}}{(d_Ad_B)^m} \\
& = \sum_{\pi\in\cS_{\vs}} d_A^{\cyc(C_m\pi)-m} d_B^{\cyc(\pi)-m}.\ea
Next, we use the fact (proved in \cite{NR-free-prob}) that for any $\pi\in\cS_m$, $\cyc(C_m\pi) + \cyc(\pi) \leq m+1$ to further bound
\be E_{d_A[d_B]}[\vs]  \leq \sum_{\pi\in\cS_{\vs}} 
d_A^{\cyc(C_m\pi)-m} d_B^{1 - \cyc(C_m\pi)}
= d_B^{1-m} \sum_{\pi\in\cS_{\vs}}  
\l(\frac{d_A}{d_B}\r)^{\cyc(C_m\pi)-m}.
\label{eq:prod-mixtures}\ee
On the other hand, if $\mu_1,\ldots,\mu_p$ are the letter frequencies of $\vs$ then \eq{Eds-perm} 
and \eq{gaussian-LB} yield
\be E_d[\vs] = \frac{\sum_{\pi\in\cS_{\vs}} d^{\cyc(C_m \pi)}}{\prod_{s=1}^p d(d+1)\cdots(d+\mu_s-1)} 
\geq  e^{-\frac{m(m+1)}{2d}} \sum_{\pi\in\cS_{\vs}} d^{\cyc(C_m \pi)-m} .
\label{eq:prod-mixtures2}\ee
Setting $d=d_A/d_B$ and combining \eq{prod-mixtures} and \eq{prod-mixtures2} yields the inequality
$$ E_{d_A[d_B]}[\vs] \leq E_{d_A/d_B}[\vs] e^{\frac{m(m+1)d_B}{2d_A}}.$$
We then raise both sides to the $k^{\text{th}}$ power and sum over $\vs$ to establish the Lemma.
\end{proof}

To avoid lengthy digressions, we avoid presenting any lower bounds for
$E_{p,d_A[d_B],k}^m$. 

Next, we also consider a model in which some of the random vectors are
repeated, which was again first proposed in \cite{LW-locking}.  Assume that  
$p^{1/k}$ is 
an
integer.  For $s=1,\ldots,p$ and $j = 1,\ldots,k$,
define 
$$s^{(j)} := \l \lceil\frac{s}{p^{1-\frac{j}{k}}} \r\rceil.$$
Note that as $s$ ranges from $1,\ldots,p$, $s^{(j)}$ ranges from
$1,\ldots,p^{j/k}$.  Define $\tilde{M}_{p,d,k} = \sum_{s=1}^p
\proj{\tilde{\varphi}_s}$, where $\ket{\tilde{\varphi}_s} =
\ket{\varphi_{s^{(1)}}^1} \ot \cdots \ot \ket{\varphi_{s^{(k)}}^k}$.  In \cite{LW-locking}, large-deviation arguments were used to show that for $x=o(1)$, $\|\tilde{M}_{p,d,k}\| = 1+o(1)$ with high probability.  Here we show that this can yield an alternate proof of our main result on the behavior of $\|M_{p,d,k}\|$, at least for small values of $x$.  In particular, we prove
\begin{corollary}\label{cor:comp-expect}
For all $m, p, d, k$,
$$ \tilde{E}_{p,d,k}^m \leq E_{p,d,k}^m.$$
\end{corollary}
This implies that if $\tilde{\lambda}$ is a randomly drawn eigenvalue of $\tilde{M}_{p,d,k}$,
 $\lambda$ is a randomly drawn eigenvalue of ${M}_{p,d,k}$ and $\gamma$ is a real number, then $\Pr[\lambda \geq \gamma] \leq \Pr[\tilde{\lambda}\geq \gamma]$.  In particular
  $$\Pr[\|M_{p,d,k}\|\geq \gamma] \leq d^k \Pr[\|\tilde{M}_{p,d,k}\| \geq \gamma].$$

The proof of \cor{comp-expect} is a direct consequence of the following Lemma, which may be of independent interest.
\begin{lemma}
If $s_i'=s_j'$ whenever $s_i=s_j$ for some strings $\vs,\vs'\in[p]^m$ then
$E_d[\vs] \leq E_d[\vs']$.
\end{lemma}
\begin{proof}
The hypothesis of the Lemma can be restated with no loss of generality
as saying that $\vs'$ is obtained from $\vs$ by a series of merges,
each of which replaces all instances of letters $a,b$ with the letter
$a$.  We will prove the inequality for a single such merge.  Next, we
rearrange $\vs$ so that the a's and b's are at the start of the
string.  This rearrangement corresponds to a permutation $\pi_0$, so
we have  
$E_d[\vs]= \tr \pi_0^\dag C_m \pi_0 \bbE [\varphi_a^{\ot \mu_a} \ot \varphi_b^{\ot \mu_b} \ot \omega]$ and
$E_d[\vs'] = \tr \pi_0^\dag C_m \pi_0 \bbE [\varphi_a^{\ot
  \mu_a+\mu_b} \ot  \omega]$, where $\omega$ is a tensor product of
various $\varphi_s$, with $s\not\in\{a,b\}$.  Taking the expectation
over $\omega$ yields a positive linear combination of various
permutations, which we absorb into the $\pi_0^\dag C_m \pi_0$ term by
using the cyclic property of the trace.  Thus we find 
\ba E_d[\vs] & = \sum_{\pi\in S_m} c_\pi \tr \pi \bbE [\varphi_a^{\ot \mu_a} \ot \varphi_b^{\ot \mu_b} \ot 
I^{m-\mu_a-\mu_b}] \\
E_d[\vs'] & = \sum_{\pi\in S_m} c_\pi \tr \pi \bbE [\varphi_a^{\ot \mu_a+\mu_b} \ot
I^{m-\mu_a-\mu_b}] \ea

for some $c_\pi \geq 0$.  A single term in the $E_d[\vs]$ sum has the
form $c_\pi \bbE[|\braket{\varphi_a}{\varphi_b}|^{2f(\pi)}]$ for some
$f(\pi)\geq 0$, while for $E_d[\vs']$, the corresponding term is
simply $c_\pi$.  Since
$\bbE[|\braket{\varphi_a}{\varphi_b}|^{2f(\pi)}]\leq 1$, this
establishes the desired inequality.
\end{proof}

\section{Approach 3: Schwinger-Dyson equations}
\label{sec:andris}
\subsection{Overview}
\label{sec:overview}

The final method we present uses the Schwinger-Dyson
equations\cite{Hastings-expander2} to evaluate traces of products of random pure
states.  First, we show how the expectation of a product of traces may be expressed as an expectation of a similar product involving fewer traces.  This will allow us to simplify $E_d[\vs]^k$, and thus to obtain a recurrence relation for $e_{p,d,k}^m$.

%We define $X_i=\ket{\varphi_i}\bra{\varphi_i}$.
%Then, we have 
%\[ \bbE[\tr M_{p,d,k}^m] = \sum_{s_1, \ldots, s_m\in\{1, \ldots, p\}}
%\bbE[\tr (X_{s_1} \ldots X_{s_m})] .\]

\subsection{Expressions involving traces of random matrices}

\subsubsection{Eliminating one ${\varphi}$: Haar random case}

We start by considering the case when $k=1$ (i.e. $\ket{\varphi_i}$ are 
just Haar-random, without a tensor product structure).
Let ${\varphi}$ be a density matrix of a Haar-random state over $\bbC^d$.

Let $A_1, \ldots, A_j$ be matrix-valued 
random variables that are independent of ${\varphi}$ (but there may be 
dependencies
between $A_i$). We would like to express 
\[ \bbE[\tr({\varphi}A_1{\varphi}A_2\ldots {\varphi}A_j)], \] 
by an expression that depends only on $A_1, \ldots, A_j$.
First, if ${\varphi}=\ket{\varphi}\bra{\varphi}$, then
\ba
\tr({\varphi}A_1{\varphi} \ldots {\varphi} A_i) \tr({\varphi} A_{i+1} {\varphi} \ldots {\varphi} A_j) 
&=
\bra{\varphi} A_1{\varphi}\ldots {\varphi} A_i \ket{\varphi} \bra{\varphi} A_{i+1} {\varphi} \ldots {\varphi} 
A_j 
\ket{\varphi}  \nn
\label{eq:merge}
&=
\bra{\varphi} A_1 {\varphi} \ldots A_i {\varphi} A_{i+1} \ldots {\varphi} A_j \ket{\varphi} =
\tr({\varphi} A_1 \ldots {\varphi} A_j) .
\ea
This allows to merge expressions that involve the same matrix ${\varphi}$.

Second, observe that ${\varphi}=U \ket{0}\bra{0}U^{\dagger}$, 
where $U$ is a random unitary 
and $\ket{0}$ is a fixed state. 
By applying eq. (19) from \cite{Hastings-expander2}, we get
\[ \bbE[\tr({\varphi}A_1{\varphi}A_2\ldots {\varphi}A_j)]  = - \frac{1}{d} \sum_{i=1}^{j-1} 
\bbE[\tr({\varphi}A_1\ldots {\varphi}A_i) 
\tr({\varphi}A_{i+1}\ldots {\varphi}A_j)] \] \[ + \frac{1}{d} \sum_{i=1}^{j} 
\bbE[\tr({\varphi}A_1\ldots A_{i-1}{\varphi}) 
\tr(A_i {\varphi}A_{i+1}\ldots {\varphi}A_j)] .\]
\comment{Since 
\[ \tr({\varphi}A_1\ldots {\varphi}A_i) \tr({\varphi}A_{i+1}\ldots {\varphi}A_p) = \tr({\varphi}A_1{\varphi}A_2\ldots 
{\varphi}A_p),\]}
Because of (\ref{eq:merge}), we can replace each term in the first sum
by $\bbE[\tr({\varphi} A_1 \ldots {\varphi} A_j)]$. Moving those terms to the left hand side
and multiplying everything by $\frac{d}{d+j-1}$ gives
\begin{equation}
\label{eq:rewrite}
 \bbE[\tr({\varphi}A_1{\varphi}A_2\ldots {\varphi}A_j)] = 
\frac{1}{d+j-1} \sum_{i=1}^{j} 
\bbE[\tr({\varphi}A_1\ldots A_{i-1}{\varphi}) 
\tr(A_i {\varphi}A_{i+1}\ldots {\varphi}A_j)] .
\end{equation}
For $i=j$, we have
\begin{equation}
\label{eq:rewrite1}
\tr({\varphi}A_1\ldots A_{j-1}{\varphi})
\tr(A_j) = \tr({\varphi}A_1\ldots A_{j-1}) \tr(A_j) .
\end{equation}
Here, we have applied $\tr(AB)=\tr(BA)$ and ${\varphi}^2={\varphi}$.
For $i<j$, we can rewrite 
\ba \tr({\varphi}A_1\ldots A_{i-1}{\varphi}) \tr(A_i {\varphi}A_{i+1}\ldots {\varphi}A_j) 
&= \tr({\varphi}A_1\ldots A_{i-1}) \tr({\varphi}A_{i+1}\ldots {\varphi}A_j A_i) 
\nn &= \label{eq:rewrite2}
\tr({\varphi}A_1\ldots A_{i-1} {\varphi}A_{i+1}\ldots {\varphi}A_j A_i).
\ea
By combining (\ref{eq:rewrite}), (\ref{eq:rewrite1}) and 
(\ref{eq:rewrite2}), we have
\ba
\bbE[\tr({\varphi}A_1{\varphi}A_2\ldots {\varphi}A_j)] 
 &= \frac{1}{d+j-1} 
 \left( \bbE[\tr({\varphi}A_1\ldots {\varphi}A_{j-1}) \tr(A_j)] +
\sum_{i=1}^{j-1} \bbE[\tr({\varphi}A_1\ldots A_{i-1} {\varphi}A_{i+1}\ldots {\varphi}A_j A_i)] 
\right)  \label{eq:onestep-exact}\\
&\leq \frac{1}{d}
 \left( \bbE[\tr({\varphi}A_1\ldots {\varphi}A_{j-1}) \tr(A_j)] +
\sum_{i=1}^{j-1} \bbE[\tr({\varphi}A_1\ldots A_{i-1} {\varphi}A_{i+1}\ldots {\varphi}A_j A_i)] 
\right) .
\label{eq:onestep}
\ea

\subsubsection{Consequences}

Consider $\bbE[\tr({\varphi}_1\ldots {\varphi}_m)]$ with ${\varphi}_i$ as described in section 
\ref{sec:overview}. Let $Y_1, \ldots, Y_l$ be the different matrix valued 
random variables that occur among ${\varphi}_1, \ldots, {\varphi}_m$.
We can use the procedure described above to eliminate all
occurrences of $Y_1$. Then, we can apply it again to eliminate all 
occurrences of $Y_2$, 
$\ldots$, $Y_{l-1}$, obtaining an expression that depends only on 
$\tr(Y_l)$.
Since $\tr(Y_l)=1$, we can then evaluate the expression. 

Each application of (\ref{eq:onestep-exact}) generates a sum of trace 
expressions with
positive real coefficients. Therefore, the final expression in $\tr(Y_l)$ 
is also a sum of
terms that involve $\tr(Y_l)$ with positive real coefficients. 
This means that $\bbE[\tr({\varphi}_1\ldots {\varphi}_m)]$ is always a positive real.  

\subsubsection{Eliminating one ${\varphi}$: the tensor product case}

We claim
\begin{lemma}
\label{lem:tensor}
Let ${\varphi}$ be a tensor product of $k$ Haar-random states in $d$ dimensions 
and $A_1, \ldots, A_j$ be matrix-valued random variables
which are independent from ${\varphi}$ and whose values are tensor products
of matrices in $d$ dimensions. Then,
\[ \bbE[\tr({\varphi}A_1{\varphi}A_2\ldots {\varphi}A_j)] \leq \\ \frac{1+j^k d^{-1/k}}{d} 
\bbE[\tr({\varphi}A_1\ldots {\varphi}A_{j-1}) \tr(A_j)] + \frac{j^k}{d^{1/k}} \sum_{i=1}^{j-1} 
\bbE[\tr({\varphi}A_1\ldots 
A_{i-1} {\varphi}A_{i+1}\ldots {\varphi}A_j A_i)]  .
\]
\end{lemma}

\proof
Because of the tensor product structure, we can express 
\[ {\varphi}={\varphi}^1\otimes {\varphi}^2\otimes \ldots {\varphi}^k, \]
\[ A_i = A^1_i \otimes A^2_i \otimes \ldots A^k_i .\]
We have 
\[ \bbE[\tr({\varphi}A_1{\varphi}A_2\ldots {\varphi}A_j)] = \prod_{l=1}^k \bbE[\tr({\varphi}^lA_1^l{\varphi}^l \ldots 
{\varphi}^lA_j^l)] .\]
We expand each of terms in the product according to \eq{onestep}.
Let $C_0=\bbE[\tr({\varphi}^l A^l_1\ldots {\varphi}^l A^l_{j-1}) \tr(A^l_j)]$ and 
\[ C_i=\bbE[\tr({\varphi}^l A^l_1\ldots 
A^l_{i-1} {\varphi}^l A^l_{i+1}\ldots {\varphi}^l A^l_j A^l_i)]\] 
for $i\in\{1, 2, \ldots, j-1\}$.
(Since each of $k$ subsystems has equal dimension $d$ and are identically distributed, 
the expectations $C_0, \ldots, C_{j-1}$ are independent of $l$.)
Then, from \eq{onestep}, we get
\[ \bbE[\tr({\varphi}A_1{\varphi}A_2\ldots {\varphi}A_j)] 
%=\prod_{l=1}^k \frac{1}{d+j-1} (C_0^l+C_1^l+\ldots+C_{j-1}^l)  \] \[
\leq
\frac{1}{d^k} \prod_{l=1}^k (C_0+C_1+\ldots+C_{j-1}) 
= \frac{1}{d^k} \sum_{i_1=0}^{j-1} \ldots \sum_{i_k=0}^{j-1}
C_{i_1} \cdots C_{i_k}.\]
Consider one term in this sum. Let $r$ be the number of $l$ for which 
$i_l=0$.
We apply the arithmetic-geometric mean inequality
\[ \frac{x_1 + x_2 + \cdots + x_k}{k} \geq \sqrt[k]{x_1 x_2 \cdots x_k} \]
to 
\[ x_l=\begin{cases} d^{-\frac{1}{k}} (C_{i_l})^k & \text{if } i_l=0 \cr
d^{\frac{r}{(k-r)k}}(C_{i_l})^k & \text{if } i_l\neq 0 \cr
\end{cases}.\]
(In cases if $r=0$ or $r=k$, we just define $x_l=C_{i_l}$ for all 
$l\in\{1, 2, \ldots, k\}$.) We now upper-bound the coefficients
of $(C_0)^k$ in the resulting sum. For $(C_0)^k$, we have
a contribution of 1 from the term which has $i_1=\ldots=i_k=0$ and
a contribution of at most $d^{-1/k}$ from every other term. 
Since there are at most $j^k$ terms, the coefficient of $(C_0)^k$
is at most 
\[ 1+ j^k d^{-1/k} .\]
The coefficient of $(C_j)^k$ in each term is at most
$d^{\frac{r}{(k-r)k}}$. Since $r\leq k-1$ (because the $r=k$ terms only 
contain
$C_0$'s), we have $d^{\frac{r}{(k-r)k}}\leq d^{\frac{k-1}{k}}$.
The Lemma now follows from there being at most $j^k$ terms.
\qed

\subsection{Main results}

\subsubsection{Haar random case}

We would like to upper-bound
\[ e_{p, d,1}^m =  \frac{1}{d} \sum_{s_1=1}^p \ldots \sum_{s_m=1}^p \bbE[\tr 
(\varphi_{s_1} \ldots \varphi_{s_m})] .\]

\begin{lemma}
\label{lem:main}
\be 
%\sum_{l=0}^{m-2} e_{p, d,1}^l e_{p-l-1, d,1}^{m-l-1} + \frac{p}{d} e_{p-1,d,1}^{m-1}\leq 
e_{p,d,1}^{m}\leq 
\sum_{l=0}^{m-2} e_{p, d,1}^{l} e_{p, d,1}^{m-l-1} + \frac{p+m^3}{d}e_{p, d,1}^{m-1} .
\ee
\end{lemma}

\proof
In section \ref{sec:main-proof}.
\qed

Using $e_{p, d,1}^0=\tr(I)/d=1$, we can state \lemref{main} equivalently as
\be 
%\sum_{l=0}^{m-1} e_{p, d,1}^l e_{p-l-1, d,1}^{m-l-1} + \l(\frac{p}{d} -1\r) 
%e_{p-1,d,1}^{m-1}\leq 
e_{p,d,1}^{m}\leq 
\sum_{l=0}^{m-1} e_{p, d,1}^{l} e_{p, d,1}^{m-l-1} + \l(\frac{p+m^3}{d}-1\r)
e_{p, d,1}^{m-1} . 
\label{eq:SD-recur1}
\ee
Define $\tilde{x}=(p+m^3)/d$ (and note that it is not exactly the same as the variable of the same name in \secref{matt}).  Then \eq{SD-recur1}
 matches the recurrence for the Narayana coefficients in \eq{Nara-recur}.   Thus we have
\begin{corollary}
\be e_{p,d,1}^m \leq \sum_{\ell=1}^m N(m,\ell) \tilde{x}^\ell = \beta_m(\tilde{x})
\leq (1+\sqrt{\tilde{x}})^{2m} \ee
\end{corollary}

Similar arguments (which we omit) establish the lower bound $e_{p,d,1}^m \geq \sum_\ell N(m,\ell)(p)_\ell / (d+m)^\ell$, which is only slightly weaker than the bound stated in \thm{trace} and proved in \lemref{complete-reduce}.

\subsubsection{Tensor product case}

The counterpart of Lemma \ref{lem:main} is

\begin{lemma}
\label{lem:main-tensor}
\ba
e_{p, d,k}^{m} & \leq 
\l(1+\frac{m^k}{d^{1/k}}\r)
\sum_{l=0}^{m-2} e_{p,d,k}^{l} e_{p,d,k}^{m-l-1} + 
\left( \frac{p}{d^k} + \frac{3 m^{k+3}}{d^{1/k}} \right) e_{p,d,k}^{m-1}  \\
 & = \l(1+\frac{m^k}{d^{1/k}}\r)
\sum_{l=0}^{m-1} e_{p,d,k}^{l} e_{p,d,k}^{m-l-1} + 
\left( \frac{p}{d^k} + 3\frac{m^{k+3}}{d^{1/k}}  - 
\l(1 + \frac{m^k}{d^{1/k}}\r) \right) e_{p,d,k}^{m-1}  \\
\ea
\end{lemma}

This time we set $\tilde{x}_k = \frac{p}{d^k} + 3\frac{m^{k+3}}{d^{1/k}}$.  Also define $\gamma = m^k / d^{1/k}$.  Then \lemref{main-tensor} implies that $e_{p,d,k}^m \leq (1+\gamma)^m [y^m] \tilde{F}(\tilde{x}_k,y)$, where $\tilde{F}$ satisfies the recurrence
\be \tilde{F} = 1 + y \tilde{F}^2 + y\l(\frac{\tilde{x}_k}{1+\gamma} - 1\r) \tilde{F}.
\label{eq:SD-gf-recur}\ee
Thus we obtain
\begin{corollary}
\ba e_{p,d,k}^m &\leq (1+\gamma)^m \beta_m\l(\frac{\tilde{x}_k}{1+\gamma}\r)
\label{eq:SD-penultimate}\\
&\leq \l( \frac{\tilde{x}_k}{x}\r)^m \beta_m(x) 
\leq \exp\l(\frac{3m^{k+4}}{xd^{1/k}}\r)\beta_m(x) 
\label{eq:SD-final-ineq}
\ea
\end{corollary}
\begin{proof}
\eq{SD-penultimate} follows from the preceding discussion as well as the relation between $\beta_m$ and the recurrence \eq{SD-gf-recur}, which was discussed in \secref{complete-reduce}  and in \cite{Nara1, Nara2}.
The first inequality in \eq{SD-final-ineq} is because $\beta_m(x(1+\eps)) \leq (1+\eps)^m\beta_m(x)$ for any $\eps\geq 0$, which in turn follows from the fact that $\beta_m(x)$ is a degree-$m$ polynomial in $x$ with nonnegative coefficients.  The second inequality follows from the inequality $1+\eps\leq e^{\eps}$.
\end{proof}
\subsection{Proofs}

\subsubsection{Proof of Lemma \ref{lem:main}}
\label{sec:main-proof}

We divide the terms $\bbE[\tr({\varphi}_{s_1}\ldots {\varphi}_{s_m})]$ into several
types.

First, we consider terms for which $s_1\notin\{s_2, \ldots, s_m\}$.
Then, ${\varphi}_{s_1}$ is independent from ${\varphi}_{s_2}\ldots {\varphi}_{s_m}$.
Because of linearity of expectation , we have
\[ \bbE[\tr({\varphi}_{s_1}\ldots {\varphi}_{s_m})]=
\tr(\bbE[{\varphi}_{s_1}] \bbE[{\varphi}_{s_2}\ldots {\varphi}_{s_m}]) = \tr \left( \frac{I}{d} 
\bbE[{\varphi}_{s_2}\ldots {\varphi}_{s_m}] \right) = \frac{1}{d} \bbE[\tr({\varphi}_{s_2} \ldots 
{\varphi}_{s_m})] .\]
By summing over all possible $s_1\in[p]$, the sum of all 
terms of this type is $\frac{p}{d}$ times
the sum of all possible $\bbE[\tr({\varphi}_{s_2} \ldots {\varphi}_{s_m})]$ with $s_1\notin 
\{s_2, \ldots, s_m\}$, i.e., $\frac{p}{d}$ times $E_{p-1, d,1}^{m-1}$.

For the other terms, we can express them as 
\begin{equation}
\label{eq:type2} 
\bbE[\tr({\varphi}_{s_1} Y_1 {\varphi}_{s_1} Y_2\ldots {\varphi}_{s_1} Y_j)] 
\end{equation}
with $Y_1, \ldots, Y_j$ being products of ${\varphi}_i$ for $i\neq s_1$.
(Some of those products may be empty, i.e. equal to $I$.)

To simplify the notation, we denote ${\varphi}={\varphi}_{s_1}$.
Because of (\ref{eq:onestep}), (\ref{eq:type2}) is less than or equal to
\begin{equation}
\label{eq:newsum}
 \frac{1}{d} \left( \sum_{i=1}^{j-1} \bbE[\tr({\varphi} Y_1 {\varphi}  
\ldots Y_{i-1} Y_i {\varphi} \ldots {\varphi} Y_j] +
 \bbE[\tr({\varphi} Y_1 {\varphi} Y_2\ldots {\varphi} Y_{j-1}) \tr(Y_j)] \right) 
\end{equation}
We handle each of the two terms in \eq{newsum} separately.
For each the term in the sum, we will upper-bound the sum of them all
(over all $\bbE[\tr({\varphi} Y_1 {\varphi} Y_2\ldots {\varphi} Y_j)]$)
by $\frac{1}{d}E_{p,d,1}^{m-1}$ times the maximum number of times
the same term can appear in the sum.

Therefore, we have to answer the question: given a term
$\bbE[\tr(Z_1 \ldots Z_{m-1})]$, what is the maximum number of ways
how this term can be generated as $\bbE[\tr({\varphi} Y_1 {\varphi}
\ldots Y_{i-1} Y_i {\varphi} \ldots {\varphi} Y_j]$?

Observe that ${\varphi}=Z_1$. Thus, given $Z_1 \ldots Z_{m-1}$,
${\varphi}$ is uniquely determined.  Furthermore, there are at most
$m$ locations in $Z_1 \ldots Z_{m-1}$ which could be the boundary
between $Y_{i-1}$ and $Y_i$. The original term 
$\bbE[\tr({\varphi} Y_1 \ldots {\varphi} Y_i)]$ can then be recovered by adding 
${\varphi}$ in that location. Thus, each term can be generated in at most $m$ 
ways
and the sum of them all is at most $\frac{m}{d} E_{p, d,1}^{m-1}$.

It remains to handle the terms of the form
\begin{equation} 
\label{eq:lasm-term}
\bbE[\tr({\varphi} Y_1 {\varphi} Y_2\ldots {\varphi} Y_{j-1}) \tr(Y_j)] .
\end{equation}
We consider two cases:

{\bf Case 1:}
There is no ${\varphi}_i$ which occurs both in $Y_j$ and in at least one
of $Y_1, \ldots, Y_{j-1}$. 
Then, the matrix valued random variables ${\varphi} Y_1 {\varphi} Y_2\ldots {\varphi} Y_{j-1}$ 
and $Y_j$ are independent.
Therefore, we can rewrite (\ref{eq:lasm-term}) as
\begin{equation}
\label{eq:lasm-term1} 
\bbE[\tr({\varphi} Y_1 {\varphi} Y_2\ldots {\varphi} Y_{j-1})] \bbE[\tr(Y_j)] .
\end{equation}
Fix $Y_1, \ldots, Y_{j-1}$.
Let $l$ be the length of $Y_j$ and let $o$ be the number of different 
${\varphi}_i$ that occur in $Y_1\ldots Y_{j-1}$.
Then, there are $p-o-1$ different ${\varphi}_i$s which can occur in $Y_j$
(i.e., all $p$ possible ${\varphi}_i$s, except for ${\varphi}_{s_1}$ and those $o$
which occur in $Y_1\ldots Y_{j-1}$).

Therefore, 
the sum of $\bbE[\tr(Y_j)]$ over all possible $Y_j$ is 
exactly $E_{p-o-1, d,1}^{l}$.
We have $E_{p-t, d,1}^{l}\leq E_{p-o-1, d,1}^{l}\leq E_{p, d,1}^{l}$.
Therefore, the sum of all terms (\ref{eq:lasm-term1}) in which $Y_j$ is 
of length $l$ is lower-bounded by the sum of all
\[ E_{p-t, d,1}^{l} \bbE[\tr({\varphi} Y_1 {\varphi} Y_2\ldots {\varphi} Y_{j-1})] \]
which is equal to $E_{p-t, d,1}^{l} E_{p, d,1}^{m-l-1}$.
Similarly, it is upper-bounded by $E_{p, d,1}^{l} E_{p,d,1}^{m-l-1}$.

{\bf Case 2:}
%$l(Y_j)>1$ and t
There exists ${\varphi}_i$ which occurs both in $Y_j$ and in 
some $Y_l$, $l\in\{1, \ldots, j-1\}$.

We express $Y_j=Z {\varphi}_i W$ and $Y_l=Z'{\varphi}_iW'$.
Then, (\ref{eq:lasm-term}) is equal to
\[ \bbE[\tr({\varphi} Y_1 \ldots Y_{l-1} {\varphi} Z {\varphi}_i W {\varphi} Y_{l+1} \ldots {\varphi} Y_{j-1} )
\tr(Z'{\varphi}_iW') =\]
\[ \bbE[\tr({\varphi} Y_1 \ldots Y_{l-1} {\varphi} Z {\varphi}_i W' Z' {\varphi}_i W {\varphi}
Y_{l+1} \ldots {\varphi} Y_{j-1} ) ] .\]
In how many different ways could this give us the same term
$\bbE[\tr(Z_1 \ldots Z_{m-1})]$?

Given $Z_1, \ldots, Z_{m-1}$, we know ${\varphi}=Z_1$. 
Furthermore, we can recover $Y_j$ by specifying the location of the first 
${\varphi}_i$, the second ${\varphi}_i$ and the location where $W'$ ends and $Z'$ 
begins.
There are at most $m-1$ choices for each of those three parameters.
Once we specify them all, we can recover the original term
(\ref{eq:lasm-term}). Therefore, the sum of all terms
(\ref{eq:lasm-term}) in this case is at most $(m-1)^3$ times the sum
of all $\bbE[\tr(Z_1 \ldots Z_{m-1})]$, which is equal to $E_{p,d,1}^{m-1}$.

Overall, we get
\begin{equation}
\label{eq:almosm-final} 
E_{p, d,1}^{m}\leq \frac{p}{d} E_{p-1, d,1}^{m-1}+\frac{m}{d} 
E_{p, d,1}^{m-1} + \sum_{l=0}^{m-2} E_{p, d,1}^{l} E_{p, d,1}^{m-l-1} 
+ \frac{(m-1)^3}{d} E_{p, d,1}^{m-1} ,
\end{equation}
with the first term coming from the terms where $s_1\notin\{s_2, \ldots, 
s_k\}$, the second term coming from the  bound on the sum in 
(\ref{eq:newsum}) and the third and the fourth terms coming from
Cases 1 and 2.
By combining the terms, we can rewrite (\ref{eq:almosm-final}) as
\be E_{p,d,1}^{m}\leq \frac{p+m^3}{d} E_{p,d,1}^{m-1} + \frac{1}{d} 
\sum_{l=0}^{m-2} E_{p,d,1}^{l} E_{p,d,1}^{m-l-1} .
\label{eq:SD-ub-pf}\ee
Dividing \eq{SD-ub-pf} by $d$  completes the proof.

We remark as well that these techniques can yield a lower bound for $E_{p,d,1}$.   To do so, we apply the inequality $1/(d+j-1)\geq 1 /(d+m)$ to \eq{onestep-exact}, and then combine the lower bounds from
the $s_1\notin\{s_2, \ldots, s_k\}$ case and Case 1.
(For the other cases, we can use 0 as the lower bound,
since we know that the expectation of any product of traces is positive.)  This yields
\be  \frac{d}{d+m}\l(
\sum_{l=0}^{m-2} e_{p, d,1}^{l} e_{p-t, d,1}^{m-l-1} + \frac{p}{d} e_{p-1,d,1}^{m-1} \r)
\leq e_{p,d,1}^{m} .
\label{eq:SD-lb-pf}\ee

\subsubsection{Proof of Lemma \ref{lem:main-tensor}}

The proof is the same as for Lemma \ref{lem:main-tensor}, except that, 
instead
of (\ref{eq:onestep}) we use Lemma \ref{lem:tensor}.

The first term in (\ref{eq:almosm-final}), $\frac{p}{d^k} E_{p-1, d,k}^{m-1}$, 
remains unchanged. The terms $\bbE[\tr ({\varphi} Y_1 \ldots Y_{i-1} Y_i {\varphi} \ldots {\varphi} 
Y_j)]$
in (\ref{eq:newsum}) 
are now multiplied by $\frac{j^k}{d^{1/k}}$ instead of $\frac{1}{d}$. 
We have $\frac{j^k}{d^{1/k}} \leq \frac{m^k}{d^{1/k}}$.
Therefore, the second term in (\ref{eq:almosm-final}) changes from $
\frac{m}{d} E_{p, d,1}^{m-1}$ to
$\frac{m^{k+1}}{d^{1/k}} E_{p,d,k}^{m-1}$. 

The terms $\bbE[\tr({\varphi} Y_1 \ldots \varphi Y_{j-1}) \tr(Y_j)]$ in (\ref{eq:newsum})
acquire an additional
factor of $1+\frac{j^k}{d^{1/k}}\leq 1+\frac{m^k}{d^{1/k}}$.
This factor is then acquired by the third and the fourth terms in 
(\ref{eq:almosm-final}).
Thus, we get 
\[ E_{p,d,k}^{m}\leq \frac{p}{d^k} E_{p-1,d,k}^{m-1}+\frac{m^{k+1}}{d^{1/k}} 
E_{p,d,k}^{m-1} + \left(1+\frac{m^k}{d^{1/k}}\right) \sum_{l=0}^{m-2} 
E_{p,d,k}^{l} E_{p, d,1}^{m-l-1} 
+ \left(1+\frac{m^k}{d^{1/k}}\right)  \frac{(m-1)^3}{d} E_{p,d,k}^{m-1} .\]

The lemma now follows from merging the second term with the fourth term. 

\subsection{Relation to combinatorial approach}
The recursive approach of this section appears on its face to be quite different from the diagrammatic and combinatorial methods discussed earlier.  However, the key recursive step in \eq{rewrite} (or equivalently \eq{onestep-exact}) can be interpreted in terms of the sorts of sums over permutations seen in \secref{aram}.

Consider an expression of the form  $X= \bbE[\tr(\varphi A_1 \varphi A_2 \ldots \varphi A_j)]$.  For the purposes of this argument, we will ignore the fact that $A_1,\ldots, A_j$ are random variables.  Letting $C_j$ denote the $j$-cycle, we can rewrite $X$ as
$$X=\tr (C_{j} \bbE[\varphi^{\ot j}] (A_1\ot  A_2 \ot \cdots \ot A_j))
= \tr (\bbE[\varphi^{\ot j}] (A_1\ot  A_2 \ot \cdots \ot A_j)),$$
since $C_j\ket{\varphi}^{\ot j} = \ket{\varphi}^{\ot j}$.
 Next we apply
\eq{Schur-average} and obtain
$$X  =
\frac{\sum_{\pi\in \cS_j} \tr (\pi(A_1\ot  A_2 \ot \cdots \ot A_j))}{d(d+1)\cdots(d+j-1)}. $$
We will depart here from the approach in \secref{aram} by rewriting
the sum over $\cS_j$.  For $1\leq i\leq j$, let $(i,j)$ denote the
permutation that exchanges positions $i$ and $j$, with $(j,j)=e$
standing for the identity permutation.  We also define
$\cS_{j-1}\subset \cS_j$ to be the subgroup of permutations of the
first $j-1$ positions.  Since $(1,j), \ldots, (j-1,j),(j,j)$ are a
complete set of coset representatives for $\cS_{j-1}$, it follows that
any $\pi\in \cS_j$ can be uniquely expressed in the form $(i,j) \pi'$
with $1\leq j$ and $\pi'\in \cS_{j-1}$.  Our expression for $X$ then
becomes 
\ban X &= \frac{1}{d+j-1} \tr\l( \sum_{i=1}^{j} (i,j) \frac{\sum_{\pi'\in \cS_{j-1}} \pi'
(A_1\ot  A_2 \ot \cdots \ot A_j)}
{d(d+1)\cdots(d+j-2)}\r). \\
& = 
\frac{1}{d+j-1} \bbE\l[\tr\l( \sum_{i=1}^{j} (i,j) (\varphi^{\ot j-1} \ot I)
(A_1\ot  A_2 \ot \cdots \ot A_j) \r)\r]
\\ & = 
\frac{1}{d+j-1} \bbE\l[\tr\l(\varphi A_1 \varphi A_2 \cdots \varphi A_{j-1}\r) \tr(A_j)
 + \sum_{i=1}^{j-1} \tr(\varphi A_1)\cdots\tr(\varphi A_{i-1})\tr(A_j\varphi A_i)
 \tr(\varphi A_{i+1})\cdots\tr(\varphi A_{j-1})\r],
\ean
which matches the expression in \eq{rewrite}, or equivalently, \eq{onestep-exact}.

The difference in approaches can then be seen as stemming from the different ways of summing over $\pi\in \cS_j$.  In \secref{aram} (and to some extent, \secref{matt}), we analyzed the entire sum by identifying leading-order terms and deriving a perturbative expansion that accounted for all the other terms.  By contrast, the approach of this section is based on reducing the sum over $\cS_j$ to a similar sum over $\cS_{j-1}$. 

\section{Lower bounds on the spectrum}
The bulk of our paper has been concerned with showing that $\|M\|$ is
unlikely to be too large (\cor{eig-LD}).  Since we give asymptotically
sharp bounds on $d^{-k} \bbE[\tr M^m]$, we in fact obtain
asymptotically convergent estimates of the eigenvalue density of $M$
(\cor{measureconverge}).  However, this does not rule out the
possibility that a single eigenvalue of $M$ might be smaller than
$(1-\sqrt{x})^2$; rather, it states that the expected number of such
eigenvalues is $o(d^k)$. 

In fact, our method was successful in proving asymptotically sharp
estimates on the largest eigenvalue of $M$.  We now turn to proving
bounds on the smallest eigenvalue of $M$.  To use the trace method to
show that w.h.p. there are {\em no} small eigenvalues, one would like
to upper bound expressions such as $\bbE[\tr(M-\lambda I)^{2m}]$, for
an appropriate choice of $\lambda$.  If we succeed in bounding such an
expression then the $\lambda_{\min}$ (the smallest eigenvalue of $M$) is
lower bounded by
\be
\bbE[(\lambda-\lambda_{\min})^2] \leq \Bigl( \bbE[\tr(M-\lambda I)^{2m}] \Bigr)^{1/m},
\ee
and hence
\be
\bbE[\lambda_{\min}] \geq  \lambda-\Bigl( \bbE[\tr(M-\lambda I)^{2m}] \Bigr)^{1/2m}.
\ee

Let us first describe a failed attempt to bound this result, before giving the correct approach.
To bound
$\bbE[\tr(M-\lambda I)^{2m}]$, the natural first attempt is to use the expansion
\be
\bbE[\tr(M-\lambda I)^{2m}]=\sum_{n=0}^{2m}  \binom{2m}{n} \bbE[\tr(M^n)]
\Bigl(-\lambda\Bigr)^{2m-n}.
\label{eq:alt-exp}
\ee
One might then attempt to estimate each term in the above expansion in turn.  Unfortunately, what happens is the following: the leading
order (rainbow) terms for $\bbE[\tr(M^n)]$ can be summed directly over $n$.  One may show that this sum contributes a result to
$\bbE[\tr(M-\lambda I)^{2m}]$ which grows roughly as ${\rm max}\{((\sqrt{x}-1)^2-\lambda)^{2m},((\sqrt{x}+1)^2-\lambda)^{2m}\}$.
That is, it is dominated by either the largest or smallest eigenvalue of the limiting distribution, depending on the value of
$\lambda$.  However, we are unable to control the corrections to this result.  While they are suppressed in powers of $1/d$, they grow
rapidly with $m$ due to the binomial factor, causing this attempt to fail.

We now describe a simple alternate approach.  Let us work within the
Feynman diagram framework.  By \eq{M-op-ineq}, the spectrum of
$M_{p,d,k}$ is close to that of $\hat{M}_{p,d,k}$ with high
probability, so we can translate bounds on $\lambda_{\min}$ in the  
Gaussian ensemble to bounds on the smallest eigenvalue in the
normalized  ensemble.

Having reduced to the Gaussian ensemble, we now construct a diagrammatic series for
$\bbE[\tr(\hat M-\lambda I)^{2m}]$.  
One way to construct such a diagrammatic series is to add in extra diagrams, in which rather
than having $m$ pairs of vertices, we instead have $n$ pairs of vertices, interspered with $m-n$ ``identity operators'', where nothing happens:
the solid lines simply proceed straight through.
However, there already is a particular contraction in our existing diagrammatic series in which
solid lines proceed straight through.  This is a particular contraction of neighboring vertices, in which
a dashed line connects the two vertices and {\it all}
vertical lines leaving the two vertices are connected to each other.  
So, we can obtain the same result by using our original diagrammatic expansion, but with a change in the rules for weighting diagrams.
If a a  diagram has a certain number, $c$, of pairs of
neighboring vertices contracted in the given way, then we adjust the weight of the diagram by
\be
\Bigl( \frac{d^{-k} p-\lambda}{d^{-k} p} \Bigr)^c
 = \frac{x-\lambda}{x}.
\ee
If $d^{-k}p-\lambda\geq 0$, then this new series consists only of positive terms and we can use our previous techniques
for estimating the series, bounding it by the sum of rainbow diagrams, plus higher order corrections.  The sum of rainbow diagrams
changes in this approach.  One could use a new set of generating functionals to evaluate the new sum of rainbow diagrams, but we
can in fact find the result more directly: we can directly use the fact that this sum is bounded by
$d^k \max\{((\sqrt{x}-1)^2-\lambda)^{2m},((\sqrt{x}+1)^2-\lambda)^{2m}\}$.  The corrections remain small.
Taking the smallest value of $\lambda$ such that
$x-\lambda\geq 0$, we have $\lambda=x$, and so we find that, for $x>1$, the sum of these diagrams is bounded by
$d^k (2\sqrt{x}+1)^{2m}$.  This gives us a bound that, for any $\epsilon>0$,
the expectation value for the smallest eigenvalue is asymptotically greater than
\be
x-2\sqrt{x}-1-\epsilon,
\ee
and hence using concentration of measure arguments and the above reduction to the Gaussian ensemble,
we can then show that, for any $\epsilon>0$, with high probability, the smallest eigenvalue of a matrix chosen randomly from the uniform ensemble is greater
than or equal to
$x-2\sqrt{x}-1-\epsilon$.

On the other hand, if $x<1$, then we will need to instead consider
$\bbE[{\rm tr}(\hat M'-\lambda I)^{2m}]$ where $\hat M'$ is the Gram matrix
of the ensemble.  Since $\hat M'$ has the same spectrum as $\hat M$
but is only $p\times p$, all of the terms in \eq{alt-exp} are
identical except that $\tr I$ equals $p$ instead of $d^k$.  We can use
a similar diagrammatic technique to incorporating the identity terms.
Now each term of $\hat{M}$ contributes the pair of vertices from
\fig{diagram}(a), but in the opposite order.  Along the horizontal,
the solid lines are the internal lines and the dashed lines are
external.  Now the identity diagrams correspond to the case when the
{\em dashed} lines proceed straight through.  These components of a
diagram initially had a contribution of 1 (with $k$ closed solid loops
canceling the natural $d^{-k}$ contribution from each pair of
vertices).  Thus, adding in the $-\lambda I$ terms results in a
multiplicative factor of $(1-\lambda)$ for each vertex pair with the
configuration where the dashed lines go straight through.  Now we can
choose $\lambda$ to be as large as 1 and still have each diagram be
nonnegative.  The resulting bound on $\bbE[\tr (\hat M' - \lambda
I)^{2m}]$ is $p
\max\{((\sqrt{x}-1)^2-1)^{2m},((\sqrt{x}+1)^2-1)^{2m}\}$ plus small
corrections.  We find that the smallest eigenvalue is $\geq
1-2\sqrt{x} - x - \eps$ with high probability.

Combining these bounds, we find that the smallest eigenvalue is
asymptotically no lower than $(1-\sqrt{x})^2 - 2\min(1,x)$.  This is
within a $1-o(1)$ factor of the unproven-but-true value of
$(1-\sqrt{x})^2$ in the limits $x\ra 0$ and $x\ra \infty$.

We believe that it should be possible to improve this result to get an
asymptotic lower bound of $(1-\sqrt x)^2$, staying within the
framework of trace methods, using any of the three techniques we have
used.  This will require a more careful estimate of the negative terms
to show that our methods remain valid.
We leave the solution of this problem to future work.

%One particular difficulty is that when $x<1$, $M_{p,d,k}$
%automatically has $(1-x)d^k$ eigenvalues equal to zero and so we need
%to distinguish between the case when these are the only low
%eigenvalues and the case when there is one additional eigenvalue below
%$(1-\sqrt{x})^2$.

\section*{Acknowledgments}
We are grateful to Guillaume Aubrun for bringing \cite{Rudelson,
  ALPT09a,ALPT09b} to our attention, for telling us about his conjecture, and for many helpful conversations on convex geometry.  AA
was supported by University of Latvia Research Grant and Marie Curie
grant QAQC (FP7-224886).  MBH was supported by U. S. DOE Contract
No. DE-AC52-06NA25396.  AWH was supported by U.S. ARO under grant
W9111NF-05-1-0294, the European Commission under Marie Curie grants
ASTQIT (FP6-022194) and QAP (IST-2005-15848), and the U.K. Engineering
and Physical Science Research Council through ``QIP IRC.''  MBH and
AWH thank the KITP for hospitality at the workshop on ```Quantum
Information Science''.


\begin{thebibliography}{10}

\bibitem{ADHW06}
A.~Abeyesinghe, I.~Devetak, P.~Hayden, and A.~Winter.
\newblock The mother of all protocols: Restructuring quantum information's
  family tree.
\newblock {\em Proc. Roc. Soc. A}, 465(2108):2537--2563, 2009.
\newblock arXiv:quant-ph/0606225.

\bibitem{ALPT09a}
R.~Adamczak, A.~Litvak, A.~Pajor, and N.~Tomczak-Jaegermann.
\newblock Quantitative estimates of the convergence of the empirical covariance
  matrix in log-concave ensembles.
\newblock {\em J. Amer. Math. Soc.}, Oct 2009.
\newblock arXiv:0903.2323.

\bibitem{ALPT09b}
R.~Adamczak, A.~E. Litvak, A.~Pajor, and N.~Tomczak-Jaegermann.
\newblock Restricted isometry property of matrices with independent columns and
  neighborly polytopes by random sampling, 2009.
\newblock arXiv:0904.4723.

\bibitem{AW02}
R.~Ahlswede and A.~Winter.
\newblock Strong converse for identification via quantum channels.
\newblock {\em IEEE Trans. Inf. Theory}, 48(3):569--579, 2002.
\newblock arXiv:quant-ph/0012127.

\bibitem{Anderson91b}
A.~Anderson, R.~C. Meyrs, and V.~Periwal.
\newblock Complex random surfaces.
\newblock {\em Phys. Lett. B}, 254(1-2):89 -- 93, 1991.

\bibitem{Anderson91a}
A.~Anderson, R.~C. Myers, and V.~Periwal.
\newblock Branched polymers from a double-scaling limit of matrix models.
\newblock {\em Nuclear Physics B}, 360(2-3):463 -- 479, 1991.

\bibitem{interlace}
R.~Arratia, B.~Bollob\'{a}s, and G.~Sorkin.
\newblock The interlace polynomial of a graph.
\newblock {\em J. Comb. Th. B,}, 92(2):199--233,, 2004.

\bibitem{BST08-expander}
A.~Ben-Aroya, O.~Schwartz, and A.~Ta-Shma.
\newblock Quantum expanders: motivation and construction.
\newblock In {\em CCC}, 2008.
\newblock arXiv:0709.0911 and arXiv:quant-ph/0702129.

\bibitem{BHLSW03}
C.~H. Bennett, P.~Hayden, D.~W. Leung, P.~W. Shor, and A.~J. Winter.
\newblock Remote preparation of quantum states.
\newblock {\em ieeeit}, 51(1):56--74, 2005.
\newblock quant-ph/0307100.

\bibitem{another-moment}
A.~Bose and A.~Sen.
\newblock Another look at the moment method for large dimensional random
  matrices.
\newblock {\em Elec. J. of Prob.}, 13(21):588--628, 2008.

\bibitem{matthias}
M.~Christandl.
\newblock {\em The structure of bipartite quantum states: Insights from group
  theory and cryptography}.
\newblock PhD thesis, University of Cambridge, 2006.
\newblock arXiv:quant-ph/0604183.

\bibitem{FZ97}
J.~Feinberg and A.~Zee.
\newblock Renormalizing rectangles and other topics in random matrix theory.
\newblock {\em J. Stat. Phys.}, 87(3--4):473--504, 1997.

\bibitem{forrester}
P.~Forrester.
\newblock Log-gases and random matrices.
\newblock unpublished manuscript.
\newblock Chapter 2. {\tt http://www.ms.unimelb.edu.au/~matpjf/matpjf.html}.

\bibitem{Hastings-expander1}
M.~B. Hastings.
\newblock Entropy and entanglement in quantum ground states.
\newblock {\em Phys. Rev. B}, 76:035114, 2007.
\newblock arXiv:cond-mat/0701055.

\bibitem{Hastings-expander2}
M.~B. Hastings.
\newblock Random unitaries give quantum expanders.
\newblock {\em Phys. Rev. A}, 76:032315, 2007.
\newblock arXiv:0706.0556.

\bibitem{Hastings-additivity}
M.~B. Hastings.
\newblock A counterexample to additivity of minimum output entropy.
\newblock {\em Nature Physics}, 5, 2009.
\newblock arXiv:0809.3972.

\bibitem{HLSW04}
P.~Hayden, P.~W.~S. D.~W.~Leung, and A.~J. Winter.
\newblock Randomizing quantum states: Constructions and applications.
\newblock {\em Comm. Math. Phys.}, 250:371--391, 2004.
\newblock arXiv:quant-ph/0307104.

\bibitem{HLW06}
P.~Hayden, D.~W. Leung, and A.~Winter.
\newblock Aspects of generic entanglement.
\newblock {\em Comm. Math. Phys.}, 265:95, 2006.
\newblock arXiv:quant-ph/0407049.

\bibitem{HW-additivity}
P.~Hayden and A.~J. Winter.
\newblock Counterexamples to the maximal p-norm multiplicativity conjecture for
  all $p > 1$.
\newblock {\em Comm. Math. Phys.}, 284(1):263--280, 2008.
\newblock arXiv:0807.4753.

\bibitem{HJ}
R.~A. Horn and C.~Johnson.
\newblock {\em Matrix Analysis}.
\newblock Cambridge University Press, 1985.

\bibitem{Johnstone-PCA}
I.~M. Johnstone.
\newblock On the distribution of the largest eigenvalue in principle components
  analysis.
\newblock {\em Annals of Statistics}, 29(2):295--327, 2001.

\bibitem{Ledoux}
M.~Ledoux.
\newblock {\em The Concentration of Measure Phenomenon}.
\newblock AMS, 2001.

\bibitem{LW-locking}
D.~W. Leung and A.~J. Winter.
\newblock {Locking 2-LOCC distillable common randomness and LOCC-accessible
  information}.
\newblock in preparation.

\bibitem{Ashley-distinguish}
A.~Montanaro.
\newblock On the distinguishability of random quantum states.
\newblock {\em Comm. Math. Phys.}, 273(3):619--636, 2007.
\newblock arXiv:quant-ph/0607011v2.

\bibitem{MP93}
R.~C. Myers and V.~Periwal.
\newblock From polymers to quantum gravity: Triple-scaling in rectangular
  random matrix models.
\newblock {\em Nuclear Physics B}, 390(3):716 -- 746, 1993.

\bibitem{NR-free-prob}
A.~Nica and R.~Speicher.
\newblock {\em Lectures on the Combinatorics of Free Probability}.
\newblock Cambridge University Press, 2006.

\bibitem{Rudelson}
M.~Rudelson.
\newblock Random vectors in the isotropic position.
\newblock {\em J. Func. Anal.}, 164(1):60--72, 1999.

\bibitem{private-super}
G.~Smith and J.~Smolin.
\newblock Extensive nonadditivity of privacy.
\newblock arXiv:0904.4050, 2009.

\bibitem{Speicher-NCPart}
R.~Speicher.
\newblock Free probability theory and non-crossing partitions.
\newblock {\em Lothar. Comb}, B39c, 1997.

\bibitem{Nara2}
R.~P. Stanley.
\newblock {\em Enumerative Combinatorics, vol. 2}.
\newblock Cambridge University Press, 1999.
\newblock Exercise 6.36 and references therein.

\bibitem{Nara1}
R.~A. Sulanke.
\newblock The {Narayana} distribution.
\newblock {\em J. of Stat. Planning and Inference}, 101(1--2):311--326, 2002.

\bibitem{Verbaar94b}
J.~Verbaarschot.
\newblock {Spectrum of the QCD Dirac operator and chiral random matrix theory}.
\newblock {\em Phys. Rev. Lett.}, 72(16):2531--2533, Apr 1994.

\bibitem{Verbaar94a}
J.~Verbaarschot.
\newblock {The spectrum of the Dirac operator near zero virtuality for $N_c =
  2$ and chiral random matrix theory}.
\newblock {\em Nuclear Physics B}, 426(3):559 -- 574, 1994.

\bibitem{YD07}
J.~Yard and I.~Devetak.
\newblock Optimal quantum source coding with quantum information at the encoder
  and decoder, 2007.
\newblock arXiv:0706.2907.

\end{thebibliography}
\end{document}